\documentclass[a4paper]{article}
\usepackage[bookmarks=true]{hyperref}

\usepackage{a4wide}
\usepackage{amsmath}
\usepackage{amsfonts}
\usepackage{amssymb}
\usepackage{amsthm}
\usepackage{bbm}
\usepackage{cleveref}
\usepackage[utf8]{inputenc}
\usepackage{color}
\usepackage{url}

\newtheorem{theorem}{Theorem}[section]
\newtheorem{definition}[theorem]{Definition}
\newtheorem{corollary}[theorem]{Corollary}
\newtheorem{lemma}[theorem]{Lemma}
\newtheorem{proposition}[theorem]{Proposition}

\newtheorem{example}[theorem]{Example}

\numberwithin{equation}{section}
\numberwithin{theorem}{section}

\newcommand{\vep}{\varepsilon}
\newcommand{\defem}[1]{{\em #1\/}}

\newcommand{\Z}{{\mathbb Z}}

\newcommand{\R}{{\mathbb R}}
\newcommand{\C}{{\mathbb C\hspace{0.05 ex}}}
\newcommand{\N}{{\mathbb N}}

\newcommand{\norm}[1]{\Vert #1\Vert}
\newcommand{\mean}[1]{\langle #1\rangle}

\newcommand{\rmd}{{\rm d}}

\begin{document}

\title{Estimation of local microcanonical averages in two lattice mean-field models using coupling techniques}
\author{Kalle Koskinen\thanks{\emailkalle} , %
Jani Lukkarinen\thanks{\emailjani}\\[1em]%
$\,^*,\,^\dag$\UHaddress}
\date{\today}
\newcommand{\email}[1]{E-mail: \tt #1}
\newcommand{\emailjani}{\email{jani.lukkarinen@helsinki.fi}}
\newcommand{\emailkalle}{\email{kalle.koskinen@helsinki.fi}}
\newcommand{\UHaddress}{\em University of Helsinki, Department of Mathematics 
and Statistics\\
\em P.O. Box 68, FI-00014 Helsingin yliopisto, Finland}

\maketitle

\begin{center}
{\it Joel L.\ Lebowitz has been one of the driving forces and main supporters of 
mathematical statistical physics for over half a century.  It is a particular 
honour and a pleasure to dedicate  this, in relation humble, update on 
foundations of statistical mechanics to him.} 
\end{center}

\bigskip

\begin{abstract}
We consider an application of probabilistic coupling techniques which provides 
explicit estimates for comparison of local expectation values between label permutation invariant states, 
for instance, between certain microcanonical, canonical, and grand canonical ensemble expectations.
A particular goal is to obtain good bounds 
for how such errors will decay with increasing system size.
As explicit examples, we focus on two well-studied
mean-field models: the discrete model of a paramagnet and the mean-field spherical model of a continuum field,
both of which are related to the Curie--Weiss model.
The proof is based on a construction of suitable probabilistic couplings between the 
relevant states, using Wasserstein fluctuation distance to control 
the difference between the expectations in the thermodynamic limit.
\end{abstract}

\tableofcontents

\section{Introduction}

We consider a novel method of analysis of convergence
of local expectation values 
in probability distributions associated with microcanonical ensembles.  Our approach aims at answering the following question which would be natural, for example, 
to control expectations in states arising in ergodic theory: Assume that
the system is in a microcanonical state with one or two known fixed conserved quantities which are label permutation invariant.
Consider an observable which depends only on a few degrees of freedom of some finite but large system, for example, consider local correlation functions.  Assume that there is some other
permutation invariant probability distribution, such as the corresponding
grand canonical ensemble, in which the expectation of the observable can be computed, either via a simulation or analytically.  Can we estimate the error which arises from replacing the typically not computable microcanonical expectation with the second result?
Assuming that the two ensembles are thermodynamically equivalent, 
how fast does the error decrease with increasing system size?

From the perspective of uniform measures with constraints, we mainly focus on the related standard ensembles, i.e., microcanonical, canonical, and 
grand canonical ensembles, each with parameters associated with the 
thermodynamically relevant quantities.  For the sake of completeness, we will 
give a heuristic overview of the standard ensemble theory in 
Sec.~\ref{sec:introtoensembles}.  There we also introduce notations and 
terminology which will be used later for defining the ensemble measures of the 
two models.

In the above standard ensemble set-up, 
the thermodynamic equivalence of ensembles can often be studied via relative entropy methods.
In certain models, in particular, of discrete lattice fields, 
the relative entropy bounds can also
provide an answer to the question stated in the beginning via the 
Pinsker inequality.  However, relative entropy estimates are not always readily available, cannot be used between measures which are not absolutely continuous at least in one direction, 
and as we will show explicitly later, the estimates they provide might not be optimal.

The motivation to look for improvements of the well developed earlier methods comes from a recent result 
\cite{Lukk18} for the supercritical Berlin--Kac spherical model \cite{berlinkac1952}. This is a model with two thermodynamic quantities in a canonical ensemble 
where one of them becomes frustrated and forms a condensate.  This results in an \defem{nonequivalence} between the canonical and grand canonical ensembles.
However, it was shown in \cite{Lukk18} that, after separating the condensate modes,
the state of the remaining modes is well described by a 
grand canonical state.  Comparing canonical with this modified grand canonical ensemble
yields local expectations which converge to each other in the thermodynamic limit.
The result was proven using 
a suitably constructed coupling and relying on the translation invariance of the system, and the resulting estimates imply a power-law convergence in the system size of the errors between the two expectation values.

The coupling technique in \cite{Lukk18} is, however, quite specific to the Berlin--Kac model, and relies partially on the existence of the condensate.  Here, we explore the extension of these ideas to well-studied cases where equivalence and non-equivalence of various ensembles are known, and which are sufficiently simple to be fairly explicitly computable.
For the models chosen here, translation invariance is being replaced by label permutation invariance, and it will serve as an important tool to lift the fairly crude 
coupling estimates into convergence of various local expectations.
We also explore the idea of replacing the standard ensembles of some of these models
with other, more accurate but still easy to evaluate, measures.
The standard theory of ensemble equivalence will serve as a guide in this choice,
and also in most of the cases studied here, it will suffice on its own.

The first of the models is the simple paramagnet which one can find in \cite{latticeexample1994}. In working with this model, there will be a slight abuse of terminology. There is no associated Hamiltonian, but the magnetization is the 
``conserved quantity'' of this model. The corresponding canonical ensemble then has a parameter associated with the control of expectation of magnetization.
The paramagnet model is, however, closely related to the standard Curie--Weiss model since the ensemble expectations of the latter can be expressed as a convex combination of those of the former (this connection will be discussed also in Sec.~\ref{sec:reltoCW}).

The second model is a continuum modification of the Curie--Weiss model called 
the mean-field spherical model.  The model has been studied in \cite{Kastner2006} and 
it 
is a simplification of the Berlin--Kac model introduced in \cite{berlinkac1952}.
In \cite{Kastner2006}, the authors consider the thermodynamic properties of the 
microcanonical and canonical ensembles. In Sec.~\ref{sec:ContCW}, we explore the 
mean-field spherical model 
in a slightly generalized set-up, namely, by also considering the density of the 
system to be a free parameter.  This allows to study the properties of the 
grand canonical ensemble 
which is not explored in \cite{Kastner2006}.
 
For both of these models, we will give detailed proofs of explicit rates of 
convergence of finite marginal distributions and/or finite moments of all order 
between the ensembles of the models.  
The main result here is the development of novel methods which employ rigorous 
and well-understood analysis of the thermodynamic properties of the ensembles in 
order to prove a form of weak convergence of the probability measures 
corresponding to the different ensembles.

For the simple paramagnet, we will supply two distinct proofs with explicit errors for the convergence of local observables. The first proof will utilize relative entropy and, as such, will mainly reprove and collect known results.  The second proof will utilize a coupling argument. Due to the simple nature of this model, we can show explicitly that the error bounds given by the coupling method are strictly better than the bounds given by the relative entropy method that we used. For the mean-field spherical model, we will focus solely on application of the coupling methods to prove local convergence results. 

The main mathematical tools for the rigorous control of expectations in the 
various ensemble measures are couplings of the ensemble measures and the 
related 
Wasserstein distance between them, with suitably chosen ``cost functions.''
We give a brief review of couplings and Wasserstein metric in 
Sec.~\ref{sec:Wassersteinandcouplings}.  

A crucial property of the ensemble measures and of the couplings constructed 
here is their  invariance under permutations of the particle labels.
The permutation invariance improves the control of differences of expectations 
under
the ensemble measures, allowing to bound the error by the above Wasserstein 
distance.
The method is similar to how translation invariance has been used in 
\cite{Lukk18} for 
the supercritical Berlin--Kac spherical model, and it is
described in detail in Sec.~\ref{sec:errorest}.  
Another tool for such an estimation is 
the Laplace method of asymptotic analysis for such integrals.  The method and 
how it applies to the above error estimation is also discussed in 
Sec.~\ref{sec:errorest} and \ref{asymptotic1}. 

We postpone more detailed discussion about further related previous works, and 
how the present estimates connect to these, at the end of Introduction, to 
Sec.~\ref{sec:relatedworks}.

\subsection{Equilibrium ensembles with 
two thermodynamic quantities}\label{sec:introtoensembles}

To fix our conventions, let us record here briefly our definitions and parametrization of the standard ensembles.
For a review of further results and discussion about the thermodynamic
equivalence of ensembles, we refer to \cite{Touchette2015} and
for mathematical details also to \cite{LJP1995}.

In the following, $\mathcal{S}$ is some arbitrary state space with 
some fixed positive \defem{reference measure} $d \phi$.  The two thermodynamic 
observables, the ``conserved quantities'',
will be called the \defem{energy} $H : \mathcal{S} \to \mathbb{R}$ 
and the \defem{particle number} $N : \mathcal{S} \to \mathbb{R}$.  We use $V > 
0$ to 
represent the \defem{number of degrees of freedom} of the system and we focus on 
the properties of the system for large $V$.
It is typically related to the ``volume'' of the state space $\mathcal{S}$ in 
some way.

We represent the constraints using, at the moment somewhat formal, 
$\delta$-function notations; the rigorous meaning of the notations will be discussed 
later.  Let us stress that we do not take the commonly used 
thin-shell smoothing of these measures
since this would for our purposes unnecessarily complicate the analysis.
However, we then have to be careful in the choice of allowed parameter 
values in some of the ensembles below, to avoid instances where the normalization
factor is zero or otherwise ill-defined.

The \defem{microcanonical ensemble} with energy density $\varepsilon \in 
\mathbb{R}$ and particle density $\rho \in \mathbb{R}$ is then given by 
\begin{align*}
\mu_{\text{MC}}^{\varepsilon, \rho;V} (d \phi)  := \frac{1}{Z_{\text{MC}} 
(\varepsilon, \rho;V)} \delta(H[\phi] - \varepsilon V) \delta (N[\phi] - \rho V) 
\ d \phi.
\end{align*}
The \defem{canonical ensemble} with inverse temperature $\beta \in \mathbb{R}$ 
and particle density $\rho \in  \mathbb{R}$ is given by
\begin{align*}
\mu_{\text{C}}^{\beta, \rho;V} (d \phi) := \frac{1}{Z_{\text{C}} (\beta, 
\rho;V)} e^{- \beta H[\phi]} \delta (N[\phi] - \rho V) \ d \phi  .
\end{align*}
Finally, the analogously defined \defem{grand canonical ensemble} with inverse 
temperature $\beta \in \mathbb{R}$ and chemical potential $\mu \in \mathbb{R}$ 
is 
\begin{align*}
\mu_{\text{GC}}^{\beta, \mu;V} (d \phi) := \frac{1}{Z_{\text{GC}} (\beta, 
\mu;V)} e^{- \beta H[\phi] - \mu N[\phi]} \ d \phi .
\end{align*}
Let us remark that, for later convenience, we do not follow the standard physics 
conventions here using which our parameter ``$\mu$'' should have been replaced by ``$- \beta\mu$''.

With the above definitions, there are a number of immediate, explicit relations between some of the above ensembles.  In particular, we will need later the following two
observations which allow representing an ensemble as a
mixture of the more constrained ensemble:
\begin{align}\label{eq:CasMCmixture}
\mu_{\text{C}}^{\beta, \rho;V} (d \phi) = \frac{1}{\int d \varepsilon \ e^{-V 
\beta \varepsilon} Z_{\text{MC}} (\varepsilon, \rho;V)} \int d \varepsilon \ 
e^{-V \beta \varepsilon} Z_{\text{MC}} (\varepsilon, \rho;V)
\mu_{\text{MC}}^{\varepsilon, \rho;V} (d \phi)
\end{align}
and
\begin{align*}
\mu_{\text{GC}}^{\beta, \mu;V} (d \phi) = \frac{1}{\int d \varepsilon d \rho \ 
e^{- V (\beta \varepsilon + \mu \rho)} Z_{\text{MC}} (\varepsilon, 
\rho;V)} \int d \varepsilon d \rho \ e^{-  V (\beta \varepsilon + \mu 
\rho)} Z_{\text{MC}} (\varepsilon, \rho;V) \mu_{\text{MC}}^{\varepsilon, \rho;V} 
(d \phi) .
\end{align*}

Next, we define the \defem{specific microcanonical entropy} or microcanonical 
entropy per degrees of freedom by
\begin{align*}
s(\varepsilon, \rho;V) := \frac{1}{V} \ln Z_{\text{MC}}(\varepsilon, \rho;V). 
\end{align*}
We define
the \defem{specific canonical free energy} or canonical free energy per degrees 
of freedom by
\begin{align*}
f_{\text{C}}(\beta, \rho;V) := - \frac{1}{V} \ln Z_{\text{C}} (\beta, \rho;V) \, 
.
\end{align*}
Note that we do not divide here by $\beta$, as would be common for definition of 
a free energy: this would not be convenient for our models since also zero and 
negative values of $\beta$ may occur here.
Similarly, the \defem{specific grand canonical free energy} or grand canonical 
free energy per degrees of freedom is defined here by
\begin{align*}
f_{\text{GC}} (\beta, \mu;V) := - \frac{1}{V} \ln Z_{\text{GC}} (\beta, \mu;V) .
\end{align*}

Now, in order to see the relationship to Laplace-type integrals, we note that 
\begin{align*}
e^{- V \beta \varepsilon} Z_{\text{MC}}(\varepsilon, \rho;V) = e^{- V (\beta 
\varepsilon - s(\varepsilon, \rho; V))}\,, \quad
e^{- V  (\beta \varepsilon + \mu \rho)} Z_{\text{MC}}(\varepsilon, \rho;V) 
= e^{- V ( \beta \varepsilon + \mu \rho - s (\varepsilon, \rho;V))} .
\end{align*}
Assuming that the limits exist, we define
\begin{align*}
s (\varepsilon, \rho) := \lim_{V \to \infty} s (\varepsilon, \rho;V), \quad 
f_{\text{C}}(\beta, \rho) = \lim_{V \to \infty} f_{\text{C}}(\beta, \rho;V), \quad 
f_{\text{GC}} (\beta, \mu) := \lim_{V \to \infty} f_{\text{GC}} (\beta, \mu;V) .
\end{align*}
Then either Laplace-type integral estimates or large deviation techniques \cite{LJP1995}
can often be used to show that the limit functions are related by a 
Legendre transform:
\begin{align*}
f_{\text{C}} (\beta, \rho) = 
\inf_\varepsilon \{ \beta \varepsilon - s (\varepsilon, \rho)\} \,, \quad
f_{\text{GC}} (\beta, \mu) = \inf_{\varepsilon, \rho} \{ \beta \varepsilon +  \mu \rho - s(\varepsilon, \rho)\} \,.
\end{align*}
Typically, 
this results in a one-to-one correspondence between the parameters 
$\varepsilon$ and $\rho$ in the microcanonical ensemble with the associated free 
parameters $\beta$ and $\mu$.  Assuming that the above thermodynamic limits 
exist and agree with each other using this correspondence, 
we say that the ensembles are \defem{thermodynamically equivalent}.

The theory of Laplace-type integrals is well-developed and allows one to compute 
explicit asymptotics of such integrals. In particular, one is typically 
interested in second-order fluctuations. Indeed, from the specific free 
energies, we obtain
\begin{align*}
\frac{\left< H \right>_{\text{C}}^{\beta, \rho ; V}}{V} = \partial_\beta 
f_\text{C} (\beta, \rho; V), \quad \frac{\left< H^2 \right>_{\text{C}}^{\beta, 
\rho ; V} - \left( \left< H \right>_{\text{C}}^{\beta, \rho ; V} \right)^2}{V} = 
- \partial_\beta^2 f_\text{C}(\beta, \rho; V) .
\end{align*}
Using the theory of Laplace-type integrals, we typically have
\begin{align*}
& \lim_{V \to \infty} \frac{\left< H \right>_{\text{C}}^{\beta, \rho ; V}}{V} = 
\partial_\beta f_{\text{C}} (\beta, \rho)\,,
\quad
\lim_{V \to \infty} \frac{\left< H^2 \right>_{\text{C}}^{\beta, \rho ; V} - 
\left( \left< H \right>_{\text{C}}^{\beta, \rho ; V} \right)^2}{V} = - 
\partial_\beta^2 f_\text{C} (\beta, \rho) .
\end{align*}
The notion of "typical" here is rather vague and we refer the reader to \ref{asymptotic1} for a more detailed account of the use of the Laplace method.  
The first limit implies that the energy density of the canonical system 
converges to a constant, which, in turn, implies that the energy density behaves 
like $O(1)$ for large $V$. The contents of the second limit imply that the 
standard deviation of the energy density of the canonical system behaves like 
$O(V^{- \frac{1}{2}})$.
However, one should not rely on these formulae directly at phase transition points where the differentiability assumptions fail: in such cases, more refined tools, such as subdifferentials and convex analysis, will be needed to study the related behaviour.

However, in addition to analysing the thermodynamic properties of the system, 
the Laplace-type analysis offers us something more. Indeed, if we return to the 
alternative representation of the canonical ensemble and we denote the 
minimizing $\varepsilon$ of $f_\text{C} (\beta, \rho)$ by $\varepsilon^*$, then, 
for some suitable class of observables $g(\phi)$, one might expect that
\begin{align}\label{eq:mcclimit}
\lim_{V \to\infty} \left| \left< g \right>_{\text{MC}}^{\varepsilon^*, \rho;V} - 
\left< g \right>_{\text{C}}^{\beta, \rho;V}\right| = 0 .
\end{align}
We then say that the two ensembles are equivalent in this observable class.  For 
instance, if the above result would hold for every function $g:\mathcal{S}\to 
\C$ which is Lipschitz continuous, we could say that the microcanonical and 
canonical ensembles are \defem{Lipschitz observable equivalent}.
Analogously, if the result holds for all polynomials $g$ of the field whose 
degree is not allowed to grow with $V$,
we say that the ensembles are \defem{equivalent in their local moments}.
In this paper, we consider the suitable class of functions $g$, and the rate of convergence in (\ref{eq:mcclimit}) in more detail.

\subsubsection{Clarification of terminology}

To avoid possible misunderstandings, let us explicitly record our usage of the terminology concerning ensembles and related objects such as the partition function and free energy.  Most notably, we will need to make a distinction between thermodynamic and auxiliary statistical ensembles. 

A statistical ensemble is a probability distribution describing the 
state of a system. A \defem{thermodynamic ensemble} is a particular 
statistical ensemble which is determined by the physical properties of the system,
in particular, by its dynamics.  The most common examples start with 
a Hamiltonian defining the dynamics and then include any other relevant conserved quantities using one of the above discussed forms leading to microcanonical,
canonical, and possibly one or more grand canonical ensembles.
Partition functions and free energies can then be associated with these
thermodynamic ensembles.  We do make some choices of convenience to simplify 
the overall constant in the partition function: to avoid misunderstandings, we include also their explicit 
definitions in the following.

Here, we start from some some given thermodynamic ensemble
in the microcanonical form.  This yields the physical probability distribution
whose local expectation values we aim to estimate.  For this estimation,
it turns out to be helpful to introduce new probability measures, i.e., statistical ensembles, on the system which we will call \defem{auxiliary ensembles}.
Since many of these auxiliary measures can be written in the same form as
standard thermodynamic ensembles, it will be helpful to extend the standard 
terminology also there, leading, for example, to ``auxiliary microcanonical ensemble with fixed magnetization density'' for the Curie--Weiss model.

The auxiliary ensembles can be associated with ``partition functions'' and ``free energies'' in analogy with the standard ensembles, and this indeed will become a helpful shorthand notation in some of our computations.
However, it should be stressed that the auxiliary ensembles usually do not have any thermodynamic meaning, for example, the magnetization defining the auxiliary ensemble above
is not implied to be a conserved quantity in any dynamics leading to the Curie--Weiss model.  In addition, when talking about phase transitions and their order, we will always refer to the parameters in the original thermodynamic ensembles, and not 
to those appearing in the auxiliary ensembles.

\subsection{Related works and further motivation}\label{sec:relatedworks}

There has always been considerable interest in trying to classify the 
``correct'' notions of convergence of the equilibrium ensembles. For a 
particularly illuminating and modern account on some of the various notions 
which have been considered, we refer to \cite{Touchette2015} and its references. 
Thermodynamic equivalence from the point of view of large deviations and convexity properties of entropy is considered in great generality in 
\cite{Touchette2015}.  Here, we approach the problem more from the point of view of convergence of generic local expectation values, and the additional facilitating ingredient is  label permutation invariance of the studied equilibrium ensembles. 
For rigorous applications of the ensembles in non-equilibrium phenomena, such as for estimating the accuracy of local thermal equilibrium while studying heat transport, it would be important to be able to estimate the error in the approximation.  This ultimate goal is the second motivation for starting with the simple example cases in the present contribution.
  
In fact, such rigorous proofs are already available in the literature, albeit 
for different systems from the ones studied here. A very detailed mathematical 
account of such a convergence has been given in \cite{Chatterjee2017} starting 
from uniform distributions on the intersection of a simplex and a sphere. By 
appropriately parametrizing the radius of the sphere, and considering the 
behaviour of finite dimensional marginals and moments of this uniform 
distribution as the dimension of the space is increased, the author is 
able to rigorously prove that a phase transition occurs for this specific 
system. In particular, the author is able to prove that in the high dimensional 
limit the finite marginal distributions of the given uniform distributions are of product form.

Another work in this direction, which cites the previous article, is given in 
\cite{Huveneers2019}. In this work the authors consider the convergence of the 
microcanonical and grand canonical measures related to the Bose--Hubbard model. 
The commonality between both \cite{Chatterjee2017} and \cite{Huveneers2019} is 
that the models they are considering are defined on state spaces with strictly 
positive unbounded elements. Such a feature seems to be a key property of these 
models since both of these works observe a phase transition into a state which 
can be characterized as containing a condensate.

In fact, a fairly satisfying account of ensembles with unbounded strictly 
positive phase spaces has been given in \cite{nam2018large}. In this work the 
author proves a form of the equivalence of ensembles for systems with multiple 
constraints satisfying certain conditions, and the results are quite general as 
to their applicability.   However, the main theorems presented there hold for phase spaces which 
are defined on $[0, \infty)^N$ rather than $\mathbb{R}^N$, and, furthermore, the 
assumptions of the main theorem do not hold for the ensembles we are 
considering here.

We also mention an extensive source for references to the relative entropy method and usage of the method in \cite{Grokinsky2008} and \cite{Chleboun2013}. 
Some of these references will also be explicitly quoted later when discussing the usage of relative entropy.

Finally, let us mention the origin of the continuum model we are considering. First, we recall the (discrete) Curie--Weiss model. For a general overview of the discrete model, we refer to \cite{Kochma_ski_2013}. We also mention the classical work of Ellis in \cite{Ellis1978} which goes beyond the standard Curie--Weiss model. In 
\cite{Kastner2006}, the authors consider a further simplification to 
the Berlin--Kac model introduced in \cite{berlinkac1952}.  In particular, the 
nearest neighbour Ising model is replaced by a mean-field Hamiltonian, and, as 
evidenced in the article, the thermodynamic properties of 
the microcanonical and canonical ensembles become exactly computable. However, 
the authors do not consider the properties of local observables in their 
analysis.  The following references contain results about 
the phase structure of these models
\cite{Kastner2006}, as well as of their Potts model type generalizations to multicomponent cases \cite{Costeniuc2005}.

Our approach differs significantly from those of the above previous works and 
their associated models. In particular, we will employ various coupling methods 
to prove convergence of finite dimensional marginal distributions and finite 
moments of all orders. In addition, our arguments do not hinge on definitions of 
the microcanonical ensembles with thin-set approximations.  Instead, we 
define the microcanonical ensembles directly as constrained measures and explore 
their properties via analytic rather than probabilistic methods. For the first model, we refer to \cite{latticeexample1994} for a considerably more detailed analysis of the various properties of the model. 
However, for the second model introduced in \cite{Kastner2006} there does not seem to be 
proofs pertaining to the convergence of finite dimensional marginals or finite 
moments. There is a considerable amount of fine structure which must be 
considered to give a full account of the local convergence at this level.

Finally, let us stress that the main purpose of this paper is to display the 
specific methods of coupling and their relationship with the local convergence 
properties of the equilibrium ensembles. The thermodynamic properties of these 
systems are already well-known and have been studied extensively, but we wish to 
give an alternate, simpler and more accurate, account of the two models present in this 
paper, with the hope that the ideas used here generalize to other, less explicitly
tractable models.

\section{Two methods of coupling and main lemmas}\label{sec:errorest}

In this section, we will present definitions relevant to this article including the concept of coupling, the Wasserstein distance metric, and their two application methods which will be presented as theorems later on.

\subsection{Couplings and Wasserstein distances}\label{sec:Wassersteinandcouplings}

We collect some of the basic notions related to couplings here. 
More thorough introduction is available for instance in \cite{Villani2009}.

\subsubsection{Couplings and transport maps}\label{sec:couplandtransport}

We will frequently make use of the notion of coupling between 
probability measures. Let $X$ be a sample space and let $\Sigma$ be a 
$\sigma$-algebra on $X$. Let $\mu_1$ and $\mu_2$ be two probability measures on 
$X$. Define the coordinate projections $P_1 : X \times X \to X$ and $P_2 : X 
\times X \to X$ by $P_1 (x,y) := x$ and $P_2 (x,y) := y$. A probability measure 
$\gamma$ on a sample space $X \times X$ with a $\sigma$-algebra $\Sigma \otimes 
\Sigma$ is called a coupling if $\gamma \circ P_1^{-1} = \mu_1$ and $\gamma 
\circ P_2^{-1} = \mu_2$.  Here, and in the following, $P^{-1}$ will  be used not only to denote the inverse of a mapping $P$, but also for the associated map which takes a set to its preimage under $P$.

In this paper, we will often give the definitions of probability measures with 
the explicit assumption that they can be constructed by simply giving suitable 
values of the expectations of measurable functions. For example, if $X$ is a 
locally compact Hausdorff space and we are able to construct a bounded positive 
linear functional $L$ on $C_c (X)$, 
the space of continuous functions with compact support equipped with the 
supremum norm, 
such that $\norm{L}=1$, 
then by the Riesz--Markov--Kakutani representation theorem, there exists a 
unique Radon probability measure\footnote{Radon measures are a subclass of Borel measures with additional technical regularity properties, cf. Wikipedia or \cite{rudin1987real}.} $\mu$ on $X$ such that $L(f) = \left< f 
\right>_\mu$ for all $f\in C_c (X)$.
 
For the contents of this paper, we will use the following equivalent notion of 
coupling. Let $f : X \to \mathbb{R}$ be a measurable function. A probability 
measure $\gamma$, as defined in the previous paragraph, is a coupling if 
\begin{align*}
\left< f \circ P_1 \right>_\gamma = \left< f \right>_{\mu_1}, \ \left< f \circ 
P_2 \right>_\gamma = \left< f \right>_{\mu_2} 
\end{align*}
holds for all such functions $f$. One typically says that the marginal 
distributions of $\gamma$ are given by $\mu_1$ and $\mu_2$.
 
In this paper, we will sometimes refer to specific types of couplings as 
transport maps. Let $\mu_1$ be a probability measure as before, and let $T : X 
\to X$ be a measurable map. Define the probability measure $\mu_2$ by setting 
$\mu_2 (A) := \mu_1 (T^{-1} (A))$ for all $A \in \Sigma$. Such a probability 
measure $\mu_2$ is called the \defem{pushfoward measure} of $\mu_1$ by the map 
$T$. We then denote $\mu_2 = T_* \mu_1$. This notion is also sometimes called 
the abstract change of variables due to the following equivalent definition of the 
pushforward measure: If $f : X \to \mathbb{R}_+$ is a 
characteristic function of a measurable set, we may set 
\begin{align}\label{eq:defmu3}
\mean{f}_{\mu_3} = \int_X \mu_1 (dx) \ f(T(x))\,,
\end{align}
and this defines a positive measure $\mu_3$ on $\Sigma$.  
Then, it is straightforward to check that 
$\mu_3$ indeed is a probability measure for which (\ref{eq:defmu3}) holds for 
every non-negative measurable function $f$.  In addition, $\mu_3=\mu_2$, and 
thus (\ref{eq:defmu3}) provides an
alternative definition of $T_* \mu_1$.

When $\mu_2$ and $\mu_1$ are measures such that there is a measurable map $T$
for which $\mu_2= T_* \mu_1$, we call $T$ a \defem{transport map} from the 
measure $\mu_1$ to $\mu_2$. A transport map $T$ can always be used to construct 
a coupling between $\mu_1$ and $\mu_2$ as follows: If $g : X \times X \to 
\mathbb{R}_+$ is a measurable function, we define a probability measure $\gamma$ 
by setting
\begin{align*}
\left< g \right>_\gamma = \int_X \mu_1 (dx) \ g(x, T(x)) .
\end{align*}
One can go through analogous steps as above and show that $\gamma$ is then indeed a 
coupling of $\mu_1$ and $\mu_2 = T_* \mu_1$.

\subsubsection{Wasserstein distance and coupling optimization}
\label{sec:Wassandoptim}

For the moment, we will specialize to probability measures on 
$\mathbb{R}^n$. Let $\mu_1$ and $\mu_2$ be probability measures on 
$\mathbb{R}^n$ and let $f : \mathbb{R}^n \to \mathbb{R}$ be a bounded 
$1$-Lipschitz function with respect to the $|| \cdot ||_p$-norm for some $p \geq 
1$.  To be explicit, we require that $f$ is a function for which 
its optimal Lipschitz constant $K$, defined by
\[
 K:=\sup_{\phi\ne \psi} \frac{|f(\phi) - 
f(\psi)|}{\norm{\phi-\psi}_{p}}\,,
\]
satisfies $K\le 1$.  This is a property which depends on the choice of norm, and restricts the class of 
allowed functions.
Naturally, if $f$ is a function with 
$K>1$, then we can apply the results below to the $1$-Lipschitz function 
$\frac{1}{K}f$, and the conclusions for the original function $f$
will be the same, as long as the constant $K$ remains bounded in $n$. 
We have chosen to use the ``$1$-Lipschitz'' assumption in order to remove one, 
otherwise quite relevant, constant from the estimates.

Suppose there exists a coupling $\gamma$ of $\mu_1$ and $\mu_2$. Using the 
properties of probability measures, we have
\begin{align}\label{eq:fromLiptonorm}
\left| \left< f \right>_{\mu_1} - \left< f \right>_{\mu_2} \right| = \left| 
\left< f \circ P_1 - f \circ P_2 \right>_\gamma\right| \leq \left< \left| f 
\circ P_1 - f \circ P_2 \right| \right>_\gamma \leq \left< || x_1 - x_2 ||_p 
\right>_\gamma .
\end{align}
On the last line, we have used the short hand notation $x_i = P_i(x)$, $i=1,2$, 
for clarity. One should note that the coupling does not appear on the left hand 
side of this inequality, and, we are thus free to minimize this inequality with 
respect to all couplings $\gamma$.  Since there always exists at least one 
coupling, given by the the product coupling $\gamma = \mu_1 \otimes \mu_2$, and 
since the functions $f$ are bounded, then for any coupling the middle expression 
has a uniform upper bound. Therefore,
\begin{align*}
\left| \left< f \right>_{\mu_1} - \left< f \right>_{\mu_2} \right| \leq 
\inf_\gamma \left< || x_1 - x_2 ||_p \right>_\gamma .
\end{align*}
Naturally, we can swap the norm $|| \cdot ||_p$ for any cost function $c(x,y) : 
\mathbb{R}^n \times \mathbb{R}^n \to \mathbb{R}_{+}$ with enough regularity as 
long as we can relate the difference of the expectations somehow to the given 
cost function.

For $p \geq 1$,
define $\mathcal{P}_p (\mathbb{R}^n)$ to be the space of probability 
measures with finite $p$:th moments, i.e., assuming that 
$\mean{\norm{x}_p^p}<\infty$.
Consider $\mu_1, \mu_2 \in 
\mathcal{P}_{p} (\mathbb{R}^n)$. 
Given also some $q\geq 1$, we denote the
$p$-Wasserstein distance between $\mu_1$ and $\mu_2$
with respect to the $q$-norm by
$W_{p;q} (\mu_1, \mu_2)$.  Explicitly, 
\begin{align*}
W_{p;q} (\mu_1, \mu_2) = \left( \inf_\gamma \int_{\mathbb{R}^n \times 
\mathbb{R}^n} \gamma (dx, dy) \ || x - y||_q^p \right)^{\frac{1}{p}} \,,
\end{align*}
and, since $\norm{x}_q\le n^{1/q}\max_j |x_j|$, is straightforward to check that 
then $W_{p;q} (\mu_1, \mu_2) <\infty$.

The $p$-Wasserstein distance has been studied comprehensively and applied in a 
great variety of circumstances; examples and discussion are provided in \cite{Villani2009}. However, for the purposes of this paper, we will 
be more interested in slightly modified cost functions which are similar in 
nature to the $p$-Wasserstein distances. The main drawback of many of the 
methods and papers associated with the Wasserstein distances is that the focus 
has been on the case where the dimension of the space $n$ is fixed. In the 
context of statistical mechanics, we are typically interested in asymptotic 
properties for arbitrarily large $n$.

\subsection{Definitions}

For the purposes of this section and for the definition of the lattice model 
later,
let us fix some shorthand notations first.
Given $N\in \N$, we denote the collection of first $N$ integers as follows
\begin{align}\label{eq:defseqN}
 [N] := \{1,2,\ldots,N\}\,,
\end{align}
and we denote the group  of permutations of its elements by $S_N$.
Given a subset $I\subset [N]$, of a length $n:=|I|$, there is 
a unique bijection $\pi_I:I\to [n]$ which retains the order of the elements in 
the 
subsequence.  We let $\bar{\pi}_I\in S_N$ denote the extension 
of $\pi_I$ which is obtained by permuting the elements in $[N]\setminus I$
in an order preserving manner into the set $[N]\setminus [n]$.
In addition, every bijection $\pi_I$ as above defines a projection $P_I:\R^N\to 
\R^n$ via the formula $(P_Ix)_j:=x_{\pi^{-1}_I(j)}$, $j\in [n]$.
Analogously, given a permutation $\pi\in S_N$, the corresponding coordinate 
permutation 
will be denoted $Q_\pi:\R^N\to \R^N$; explicitly, we set $(Q_\pi 
x)_j:=x_{\pi^{-1}(j)}$, $j\in [N]$ (note that using the inverse permutation in 
the formula will result in a map which will send coordinate $i$ into coordinate 
$\pi(i)$).

Given $y\in \R$, there is a unique integer $k\in \Z$ for which 
$k \leq y < k + 1$, and we denote this by using the ``floor'' notation, 
$k:=\lfloor y \rfloor$.
In particular, given $n,N\in \N$ such that $n\le N$ and setting
$k=\lfloor N/n \rfloor$ we have $k\in \N$ and $k$ satisfies $k n\leq N < (k + 1) 
n$.

\begin{definition}[Permutation invariance of measures on $\R^N$]
Given $N\in \N$, a probability measure $\mu$ on $\R^N$, we say that $\mu$ is 
permutation invariant, if for every integrable function $f:\R^N\to \R$ and 
a permutation $\pi\in S_N$, we have $f\circ Q_\pi \in L^1(\mu)$ and
\[
 \mean{f\circ Q_\pi}_\mu = \mean{f}_\mu\,.
\]
\end{definition}

Finally, instead of using a standard $p$-norm to measure distances in $\R^N$, 
we scale it suitably with $N$ so that the Wasserstein cost function becomes an 
average over particle labels.  The benefits of this definition will become 
apparent in Sec.~\ref{sec:directcoupling}.
\begin{definition}[Specific $p$-norm fluctuation distance]  Suppose $p \geq 1$ 
and $N\in \N$. Let $\mu_1$ and $\mu_2$ be two Radon probability measures on 
$\mathbb{R}^N$ such that the $p$:th  moments under both measures are finite. 
Their specific $p$-norm fluctuation distance $w_p$ is then defined as
\begin{align*}
w_p (\mu_1, \mu_2;N) :=  \left(\inf_{\gamma} \int_{\mathbb{R}^N \times 
\mathbb{R}^N} \gamma (dx, dy) \frac{1}{N}\sum_{i=1}^N |x_i - y_i|^p 
\right)^{\frac{1}{p}}\, ,
\end{align*}
where the infimum is taken over all couplings of $\mu_1$ and $\mu_2$. 
\end{definition}
Clearly, this definition relates to the standard $p$-norm Wasserstein distance 
mentioned earlier via a scaling: $w_p = N^{-\frac{1}{p}} W_{p;p}$.

\subsection{The direct coupling method} \label{sec:directcoupling}

To highlight the benefits of the above definitions, we 
offer the following fundamental Lemma which will be used to prove the main 
theorems of this paper.  It should be stressed that the key assumption is to 
specialize to permutation invariant measures.  We aim to consider local 
expectations, i.e., $\mean{F}$
for functions $F:\R^N\to \R$ which depend only on components $x_i$, $i\in I$,
where $I\subset [N]$ can be otherwise arbitrary but it has a bounded size, i.e., 
$|I|$ remains bounded when $N\to \infty$.  In particular, note that then there 
is some $f: \mathbb{R}^{|I|} \to \mathbb{R}$ such that $F=f \circ P_{I}$.

\begin{lemma} \label{wass2est}
Suppose $p \geq 1$ and $N\in \N$.
Let $\mu_1$ and $\mu_2$ be two permutation invariant Radon probability measures 
on $\mathbb{R}^N$ such that the $p$:th moments under both measures are finite.
Consider a subset $I\subset[N]$.
If $f : \mathbb{R}^{|I|} \to \mathbb{R}$ is a bounded $1$-Lipschitz function 
with respect to the $|| \cdot ||_p$ norm, then we have
\begin{align*}
\left| \left< f \circ P_{I} \right>_{\mu_1} - \left< f \circ P_{I} 
\right>_{\mu_2} \right| \leq  \left( \frac{|I|}{1 - \frac{|I|}{N}} 
\right)^{\frac{1}{p}} w_p (\mu_1, \mu_2;N) \,.
\end{align*}
\end{lemma}
\begin{proof}
For the proof, set $n:=|I|$ and 
$k:=\lfloor N/n \rfloor$ when $k\in \N$ and $k$ satisfies $k n\leq N < (k + 1) 
n$.
We define the sets $I_i \subset [N]$, $i\in [k]$, by setting $I_1 := I$ and, for 
$i > 1$, we proceed inductively by selecting $|I|$ elements from  the set $[N] 
\setminus \left( \bigcup_{j=1}^{i-1} I_j \right)$ to be the set $I_i$. The 
collection of sets $I_i$ are disjoint and $\bigcup_{i=1}^k I_i \subset [N]$. 
For any $i$, there is a permutation in $S_N$ which is bijection between $I_i$ 
and $I$.
Thus by the assumed  permutation invariance of the measures, we have $\left< f 
\circ P_{I_i} \right> =  \left< f \circ P_{I} \right>$ for either measure and 
all $i$. 
Therefore,
\begin{align*}
\left< f \circ P_{I} \right>_{\mu_1} - \left< f \circ P_{I} \right>_{\mu_2} 
= \frac{1}{k} \sum_{i=1}^k \left( \left< f \circ P_{I_i} \right>_{\mu_1} - 
\left< f \circ P_{I_i} \right>_{\mu_2} \right)\,.
\end{align*}

Suppose then that $\gamma$ is a coupling between $\mu_1$ and $\mu_2$.
Then $\left< f \circ P_{I_i} \right>_{\mu_j}=
\left< f \circ P_{I_i} \circ P_j\right>_{\gamma}$ for both $j=1,2$.
Again resorting to the shorthand notations $x_j := P_jx$, we can rewrite
\begin{align*}
\left< f \circ P_{I_i} \right>_{\mu_1} - \left< f \circ P_{I_i} \right>_{\mu_2} 
= \left< f(P_{I_i}x_1)-f(P_{I_i}x_2)\right>_{\gamma} \,.
\end{align*}
The absolute value of this expression can now be estimated using the 
assumed $1$-Lipschitz property of $f$.  Combining the results and using the 
triangle inequality we thus obtain
\begin{align*}
\left| \left< f \circ P_{I} \right>_{\mu_1} - \left< f \circ P_{I} 
\right>_{\mu_2}\right| \leq \frac{1}{k} \sum_{i=1}^k 
\mean{\norm{P_{I_i}x_1-P_{I_i}x_2}_p}_\gamma
\leq \left(\frac{1}{k} \sum_{i=1}^k 
\mean{\norm{P_{I_i}x_1-P_{I_i}x_2}^p_p}_\gamma\right)^{\frac{1}{p}}
\,,\end{align*}
where in the last step we have used H\"older's inequality.
Since the sets $I_i$ are disjoint, here
$\sum_{i=1}^k \norm{P_{I_i}x_1-P_{I_i}x_2}^p_p\le \sum_{j=1}^N 
|(x_1)_j-(x_2)_j|^p$. 
Therefore,
\begin{align*}
\left| \left< f \circ P_{I} \right>_{\mu_1} - \left< f \circ P_{I} 
\right>_{\mu_2} \right| \leq  \left( \frac{1}{k}\left< || x_1 - x_2 ||_p^p 
\right>_\gamma \right)^\frac{1}{p}.
\end{align*}
Because the left hand side of the above estimate does not depend on the coupling 
 $\gamma$, we can take the infimum over all possible couplings.  Then using the 
relation between $k$ and $n$ stated in the beginning of the proof, we obtain
\begin{align*}
\left| \left< f \circ P_{I} \right>_{\mu_1} - \left< f \circ P_{I} 
\right>_{\mu_2} \right| \leq \left( \frac{n}{1 - \frac{n}{N}}  
\right)^{\frac{1}{p}} w_p (\mu_1, \mu_2;N)\,,
\end{align*}
as desired.
\end{proof}

The first Lemma concerned bounds on local observables which were bounded 
$1$-Lipschitz functions. This next variant of the Lemma concerns 
estimation of arbitrary finite moments.
\begin{theorem} \label{wass2momest}
Suppose $p > 1$ and $N\in \N$.
Let $\mu_1$ and $\mu_2$ be two permutation invariant Radon probability measures 
on $\mathbb{R}^N$ such that the $p_0$:th moments of both measures are finite
for some $p_0\ge p$.
 
Let $J$ be a finite sequence of elements in $[N]$ where elements may be 
repeated. 
Let $n_J:=|J|$ denote the length of the sequence and $I\subset [N]$
the collection of elements occurring in the sequence, i.e., set $I:=\{J_\ell 
{\,|\,} \ell \in [n_J]\}$.  For any $x \in \mathbb{R}^N$, we then let $x^J$ 
denote the power
\begin{align*}
x^J := \prod_{\ell = 1}^{n_J} x_{J_\ell} .
\end{align*}

Assuming also $n_J\le p_0+1-\frac{p_0}{p}$, it follows that
\begin{align*}
\left|  \left< x^J \right>_{\mu_1} - \left< x^J \right>_{\mu_2} \right| \leq n_J 
M(J,p)^{n_J-1} \left( \frac{|I|}{1 - \frac{|I|}{N}} \right)^{\frac{1}{p}} w_p 
(\mu_1, \mu_2;N) \,,
\end{align*}
where $M(J,p)=1$ if $n_J=1$, and otherwise
\begin{align*}
M(J,p) := \max_{i\in I}\left(\left< 
|x_{i}|^{q(n_J-1)}\right>^{\frac{1}{q(n_J-1)}}_{\mu_1},\left< 
|x_{i}|^{q(n_J-1)}\right>^{\frac{1}{q(n_J-1)}}_{\mu_2}\right)<\infty 
\,,
\end{align*}
with $q = \frac{p}{p - 1}$.
\end{theorem}
\begin{proof}
Generalized H\"older's inequality implies that
\[
 \mean{|x^J|}\le \prod_{\ell = 1}^{n_J} 
\mean{|x_{J_\ell}|^{n_J}}^{\frac{1}{n_J}}\,.
\]
Since $1\le n_J\le p_0$, the assumptions guarantee that $x^J$ is integrable with 
respect to both $\mu_1$ and $\mu_2$.  On the other hand, $n_J\le 
p_0+1-\frac{p_0}{p}$
implies $q(n_J-1)\le p_0$, so that also $M(J,p)<\infty$.

First, note that for $x,y \in \mathbb{R}^N$, we have
\begin{align*}
x^J - y^J = \sum_{i=1}^{n_J} (x_{J_i} - y_{J_i}) \prod_{j < i} x_{J_j} \prod_{k 
> i} y_{J_k} \,.
\end{align*}
There are $n_J$ factors in each of the products under the sum.  Thus for any 
coupling $\gamma$ between $\mu_1$ and $\mu_2$ and, for simplicity, replacing
$x_1,x_2$ by $x,y$, we find using the generalized H\"older's inequality
\begin{align*}
& \left| \left< x^J - y^J \right>_\gamma \right|  \leq 
\sum_{i=1}^{n_J} \left< |x_{J_i} - y_{J_i}| \prod_{j < i} |x_{J_j}| \prod_{k > 
i} |y_{J_k}| \right>_\gamma \\ & \quad
 \leq \sum_{i=1}^{n_J} 
  \left< |x_{J_i} - y_{J_i}|^p \right>_\gamma^{\frac{1}{p}}
  \prod_{j < i} \left< |x_{J_j}|^{q'} \right>_\gamma^{\frac{1}{q'}}
  \prod_{k> i} \left< |y_{J_k}|^{q'} \right>_\gamma^{\frac{1}{q'}}
 \, ,
\end{align*}
where $q' := q(n_J-1)$, so that indeed $\frac{1}{p}+(n_J-1)\frac{1}{q'}=1$, as 
required by the H\"older's inequality.  Apart from the first term, the remaining 
$n_J-1$ terms
are all bounded by $M(J,p)$.  Therefore,
\begin{align*}
& \left| \left< x^J - y^J \right>_\gamma \right|  
 \leq M(J,p)^{n_J-1}\sum_{i=1}^{n_J} 
  \left< |x_{J_i} - y_{J_i}|^p \right>_\gamma^{\frac{1}{p}}
  \le M(J,p)^{n_J-1} n_J \left(\frac{1}{n_J}\sum_{i=1}^{n_J} 
  \left< |x_{J_i} - y_{J_i}|^p \right>_\gamma\right)^{\frac{1}{p}}
 \, ,
\end{align*}
where H\"older's inequality has been used in the second step.
Here, even if there are repetitions in the sequence $J$,
we have $\frac{1}{n_J}\sum_{i=1}^{n_J}  |x_{J_i} - y_{J_i}|^p 
\le \sum_{i\in I} |x_{i} - y_{i}|^p =: \norm{x-y}_{I,p}^{p}$.  Therefore,
\begin{align*}
& \left| \left< x^J - y^J \right>_\gamma \right|  
 \leq n_J M(J,p)^{n_J-1}
  \left<\norm{x-y}_{I,p}^{p} \right>_\gamma^{\frac{1}{p}}
 \, .
\end{align*}

To finish the proof, one should notice that the label subset $I$ which appears 
in this theorem can be regarded in the same way as in the proof of Lemma 
\ref{wass2est}. 
Using the assumed permutation invariance to clone the labels yields collections 
of subsequences $J(\ell)$ and subsets $I_\ell$ for $\ell\in [k]$, 
where $k:=\lfloor N/n \rfloor$, $n:=|I|$.   Since $\mean{x^J}_{\mu_i} = 
\mean{x^{J(\ell)}}_{\mu_i}$ by construction, we find using permutation 
invariance that 
\begin{align*}
& \left| \left< x^J \right>_{\mu_1} - \left< x^J \right>_{\mu_2} \right| \leq 
\frac{1}{k} \sum_{\ell=1}^k \left| \left< x^{J(\ell)} - y^{J(\ell)} 
\right>_\gamma \right|
\leq \frac{n_J M(J,p)^{n_J-1}}{k} \sum_{\ell=1}^k \left< || x - y 
||_{I_\ell,p}^p\right>_\gamma^\frac{1}{p} \\
&\quad \leq n_J M(J,p)^{n_J-1} \left( \frac{|I|}{1 - \frac{|I|}{N}} 
\right)^{\frac{1}{p}}  w_p (\mu_1, \mu_2;N)\,,
\end{align*}
as desired.
\end{proof}

\subsection{Free energy method}

By applying Lemma \ref{wass2est}, we are now able to produce two distinct types of 
coupling proofs which concern the ensembles discussed in the introduction.
\begin{theorem} \label{freeenergylip}
Let $\mu_{\text{MC}}^{\varepsilon, \rho;N}$ be a permutation invariant 
probability measure corresponding to a microcanonical ensemble with energy 
density $\varepsilon$ and particle density $\rho$. In addition, assume that if 
we fix a possible energy density $\varepsilon'$, then for any other possible 
energy density $\varepsilon$ there exists a constant $C(\varepsilon, \rho) > 0$ 
independent of $\varepsilon'$ and $N$, but possibly dependent on $\varepsilon$ 
and $\rho$, such that
\begin{align*}
w_p (\mu_{\text{MC}}^{\varepsilon, \rho;N}, \mu_{\text{MC}}^{\varepsilon', 
\rho;N};N) \leq C(\varepsilon, \rho) |\varepsilon - \varepsilon'|
\end{align*}
for some $p \geq 1$.  Suppose also that the microcanonical and canonical 
measures, for some parameter $\beta$, have finite $p$:th moments.
 
Fix $n < \infty$ and consider any $I\subset\mathbb{N}$ of length $n$.
Let $f :\mathbb{R}^{|I|} \to \mathbb{R}$ be a bounded $1$-Lipschitz function 
with respect to the $|| \cdot ||_p$ norm. Then
\begin{align*}
\left| \left< f \circ P_I \right>_{\text{MC}}^{\varepsilon, \rho;N} - \left<f 
\circ P_I \right>_{\text{C}}^{\beta, \rho ;N} \right| \leq C(\varepsilon, \rho) 
\left( \frac{ |I|}{1 - \frac{|I|}{N}} \right)^{\frac{1}{p}} \left( 
\sigma_{\text{C}}^{\beta, \rho;N} \!\left( \frac{H}{N} \right) + \left| 
\varepsilon - \frac{\left< H \right>_{\text{C}}^{\beta, \rho;N}}{N} 
\right|\right),
\end{align*}
where the canonical standard deviation of energy density reads explicitly
\begin{align*}
\sigma_{\text{C}}^{\beta, \rho;N}\!\left( \frac{H}{N} \right) = 
\sqrt{\frac{\left< H^2 \right>_{\text{C}}^{\beta, \rho ; N} - \left( \left< H 
\right>_{\text{C}}^{\beta, \rho ; N} \right)^2}{N^2}}\, .
\end{align*}
Using the notation of the specific free energies, the same result can be 
rewritten as
\begin{align*}
\left| \left< f \circ P_I \right>_{\text{MC}}^{\varepsilon, \rho;N} - \left<f 
\circ P_I \right>_{\text{C}}^{\beta, \rho ;N} \right| \leq C(\varepsilon, \rho) 
\left( \frac{ |I|}{1 - \frac{|I|}{N}} \right)^{\frac{1}{p}}  \left( 
\frac{1}{\sqrt{N}} \sqrt{- \partial_\beta^2 f_{\text{C}}(\beta, \rho;N)} + 
\left| \varepsilon - \partial_\beta f_\text{C} (\beta, \rho;N) \right|\right)
\end{align*}
\end{theorem}
\begin{proof}
By the relation (\ref{eq:CasMCmixture}), we have
\begin{align*}
&\left< f \circ P_I \right>_{\text{MC}}^{\varepsilon, \rho;N} - \left<f \circ 
P_I \right>_{\text{C}}^{\beta, \rho ;N} 
\\ & \quad =  
 \frac{1}{\int d \varepsilon' \ e^{-N \beta \varepsilon'} Z_{\text{MC}} 
(\varepsilon', \rho;N)} \int d \varepsilon' \ e^{-N \beta \varepsilon'} 
Z_{\text{MC}} (\varepsilon', \rho;N)  \left( \left< f  \circ P_I 
\right>_{\text{MC}}^{\varepsilon, \rho;N} - \left< f \circ P_I  
\right>_{\text{MC}}^{\varepsilon', \rho;N} \right) \,. 
\end{align*}
Applying \Cref{wass2est} together with the assumptions of this theorem, we thus 
obtain
\begin{align*}
&\left| \left< f \circ P_I \right>_{\text{MC}}^{\varepsilon, \rho;N} - \left<f 
\circ P_I \right>_{\text{C}}^{\beta, \rho ;N} \right| \\
&\quad \leq  \frac{1}{\int d \varepsilon' \ e^{-N \beta \varepsilon'} 
Z_{\text{MC}} (\varepsilon', \rho;N)} \int d \varepsilon' \ e^{-N \beta 
\varepsilon'} Z_{\text{MC}} (\varepsilon', \rho;N)  \left( \frac{|I| }{1 - 
\frac{|I|}{N}} \right)^{\frac{1}{p}}
C(\varepsilon,\rho)|\varepsilon - \varepsilon'| \,.
\end{align*}
Since
\[
 |\varepsilon - \varepsilon'| \le
 \left|\varepsilon - \frac{\left< H \right>_{\text{C}}^{\beta, 
\rho;N}}{N}\right| 
 + \left| \frac{\left< H \right>_{\text{C}}^{\beta, \rho;N}}{N}- 
\varepsilon'\right| \, ,
\]
where the first term on the right hand side does not depend on $\vep'$, we obtain by H\"older's 
inequality
an estimate
\begin{align*}
&\left| \left< f \circ P_I \right>_{\text{MC}}^{\varepsilon, \rho;N} - \left<f 
\circ P_I \right>_{\text{C}}^{\beta, \rho ;N} \right| \\
& \quad \leq  C(\varepsilon,\rho)\left( \frac{|I|}{1 - \frac{|I|}{N}} 
\right)^{\frac{1}{p}} \Biggl[\,\Biggl|\varepsilon -  \frac{\left< H 
\right>_{\text{C}}^{\beta, \rho;N}}{N}\biggr| 
\\ & \qquad 
+ \Biggl( \frac{1}{\int d \varepsilon' \ e^{-N \beta \varepsilon'} Z_{\text{MC}} 
(\varepsilon', \rho;N)} \int d \varepsilon' \ e^{-N \beta \varepsilon'} 
Z_{\text{MC}} (\varepsilon', \rho;N)  \left|\frac{\left< H 
\right>_{\text{C}}^{\beta, \rho;N}}{N} - \varepsilon' \right|^2 
\Biggr)^{\frac{1}{2}}  \Biggr]\\
&= C(\varepsilon, \rho)\left( \frac{ |I|}{1 - \frac{|I|}{N}} 
\right)^{\frac{1}{p}} \left( \sigma_{\text{C}}^{\beta, \rho;N}\! \left( 
\frac{H}{N} \right) + \left| \varepsilon - \frac{\left< H 
\right>_{\text{C}}^{\beta, \rho;N}}{N} \right|\right),
\end{align*}
as desired.  Then we use the generic properties listed in 
Sec.~\ref{sec:introtoensembles} to express the result in terms of the canonical 
free energy.
\end{proof}

Following the theme of the direct coupling method, the approach can also then be 
applied to the case of finite moments.
\begin{theorem}
Let $\mu_{\text{MC}}^{\varepsilon, \rho;N}$ be a permutation invariant 
probability measure corresponding to a microcanonical ensemble with energy 
density $\varepsilon$ and particle density $\rho$. In addition, assume that if 
we fix a possible energy density $\varepsilon'$ then for any other possible 
energy density $\varepsilon$ there exists a constant $C(\varepsilon, \rho) > 0$ 
independent of $\varepsilon'$ and $N$, but possibly dependent on $\varepsilon$ 
and $\rho$ such that
\begin{align*}
w_p (\mu_{\text{MC}}^{\varepsilon, \rho;N}, \mu_{\text{MC}}^{\varepsilon', 
\rho;N};N) \leq C(\varepsilon, \rho) |\varepsilon - \varepsilon'|
\end{align*}
for some $p > 1$. Suppose also that the microcanonical and canonical measures, 
for some parameter $\beta$, have finite $p_0$:th moments for some $p_0\ge p$.

Let $J$ be a finite sequence of elements in $[N]$ where elements may be 
repeated,
let $n_J:=|J|$, and suppose that $n_J\le p_0+1-\frac{p_0}{p}$.
Collect into $I\subset [N]$ the elements occurring in the sequence.
It follows that
\begin{align*}
\left| \left< \phi^J \right>_{\text{MC}}^{\varepsilon, \rho; N} - \left< \phi^J 
\right>_{\text{C}}^{\beta, \rho;N} \right| \leq 
C(\varepsilon, \rho) n_J M(J,p)^{n_J-1} \left( \frac{|I|}{1 - \frac{|I|}{N}} 
\right)^{\frac{1}{p}}
 \left( \sigma_{\text{C}}^{\beta, \rho;N}\! \left( \frac{H}{N} \right) + \left| 
\varepsilon - \frac{\left< H \right>_{\text{C}}^{\beta, \rho;N}}{N} 
\right|\right),
\end{align*}
where, using the dual exponent $q = \frac{p}{p - 1}$,
\begin{align*}
M(J,p) := \max_{i\in I}\left(\left< 
|\phi_{i}|^{q(n_J-1)}\right>_{\text{MC}}^{\varepsilon, \rho; N},\left< 
|\phi_{i}|^{q(n_J-1)}\right>_{\text{MC}}^{\varepsilon', \rho; N}\right)^{\frac{1}{q(n_J-1)}}<\infty 
\,.
\end{align*}
\end{theorem}
\begin{proof}
The proof is almost identical to the proof of the previous theorem. In order to 
isolate the moments of the canonical ensemble, one needs an additional 
application of H\"older's inequality. 
\end{proof}

For suitable ensembles, these theorems together imply that with bounded moments, 
one can achieve an explicit rate of convergence of the finite dimensional 
moments and marginals of the ensembles. 

\section{Simple model of a paramagnet}

In this section, we will consider a simple model of a paramagnet discussed in \cite{latticeexample1994}. We begin by defining the magnetization of the lattice system and the two probability measures on the lattice that we will consider.

\begin{definition} 
Let $\Lambda$ be a finite lattice with size $N := |\Lambda|$. We define the space of spin configurations $\mathcal{S} := \{ -1 , 1\}^{\Lambda}$. Define the magnetization $M : \mathcal{S} \to \mathbb{R}$ by
\begin{align*}
M[\phi] := \sum_{x \in \Lambda} \phi(x) \,.
\end{align*}
Furthermore, we define the magnetization density $m_N : \mathcal{S} \to \mathbb{R}$ by $m_N[\phi] = \frac{M[\phi]}{N}$.
\end{definition}

In the following definitions, there will be a slight misuse of the ensemble terminology. We will refer to the fixed magnetization probability measures as an auxiliary microcanonical ensemble, and the probability measure with a parameter controlling the expectation of magnetization will be referred to as an auxiliary canonical ensemble.

\begin{definition}[Fixed magnetization ensemble/auxiliary microcanonical ensemble]\label{th:defparammicro} Let $m  \in \text{Ran}[m_N]$. For such an $m$, define the set 
$\mathcal{S}_m := \{ \phi \in \mathcal{S} : m_N[\phi] = m \}$. The 
auxiliary microcanonical ensemble with magnetization density $m$ is defined via its action 
on functions $f : \mathcal{S} \to \mathbb{R}$ by
\begin{align*}
\left< f \right>_{\text{MC}}^{m;N} := \frac{1}{|\mathcal{S}_m|} \sum_{\phi \in 
\mathcal{S}_m} f (\phi) .
\end{align*}
\end{definition}

\begin{definition} [Fluctuating magnetization/auxiliary canonical ensemble] Let 
$\mu \in \mathbb{R}$. The auxiliary canonical ensemble with magnetic potential $\mu$ is 
defined via its action functions $f : \mathcal{S} \to 
\mathbb{R}$ by
\begin{align*}
 \left< f \right>_{\text{C}}^{\mu;N} &:= \frac{1}{\sum_{\phi \in \mathcal{S}} 
e^{- \mu M [\phi]}} \sum_{\phi \in \mathcal{S}} e^{- \mu M[\phi]} f(\phi) 
\\ & \quad
= \frac{1}{\sum_{m \in \text{Ran}[m_N]} e^{- \mu m N}|\mathcal{S}_m|} \sum_{m 
\in \text{Ran}[m_N]} e^{- \mu m N}|\mathcal{S}_m| \left< f 
\right>_{\text{MC}}^{m ;N} .
\end{align*}
The second representation is called the magnetization representation 
of this ensemble. 
\end{definition}

We will also need the following standard ``thermodynamic'' properties of these ensembles. We have compiled them in the following lemma.
\begin{lemma} \label{partitionmag} For $m \in \operatorname{Ran}[m_N] \setminus \{ -1 , 1 \}$, the partition function of the fixed magnetization ensemble is given by
\begin{align*}
Z_{\operatorname{MC}} (m;N) := |\mathcal{S}_m| = \binom{N}{\frac{1 + m}{2} N} = \left( (N + 1) \int_0^1 dt \ e^{N \left( \frac{1 + m}{2} \ln t + \frac{1 - m}{2} \ln (1 - t) \right)}\right)^{-1} .
\end{align*}
The partition function of the fluctuating 
magnetization ensemble is given by
\begin{align*}
Z_{\text{C}} (\mu;N) := \sum_{\phi \in \mathcal{S}} e^{-\mu M[\phi]} = \left(2 
\cosh (\mu) \right)^N,
\end{align*}
and the specific free energy is given by
\begin{align*}
f_{\text{C}} (\mu;N) := - \frac{1}{N} \ln Z_{\text{C}} (\mu;N) = - \ln (2 \cosh 
(\mu)) .
\end{align*}
The average and standard deviation of the magnetization density are given by
\begin{align*}
\frac{\left< M \right>_{\text{C}}^{\mu;N}}{N} = - \tanh (\mu) ,
\end{align*}
and
\begin{align*}
\sigma_{\text{C}}^{\mu;N} \left( \frac{M}{N} \right)  =  \frac{1}{\sqrt{N}} 
\sqrt{1 - \tanh^2 (\mu)} .
\end{align*}
\begin{proof}
The calculation of the auxiliary microcanonical partition function is based on the fact that the number of positive spins in a field configuration fully defines the total magnetization of the configuration, and, as a result, one only needs to consider the number of configurations with a specific number of positive spins. The final equality follows from the representation of the beta function after opening up the combination and subsequent factorials.

The rest of the results concerning the auxiliary canonical ensemble follow by first computing the partition function, by noting that the structure of the measures is that of a product measure.  Then we can differentiate the free energy with 
respect to $\mu$ and divide appropriately by the degrees of freedom $N$. 
\end{proof}
\end{lemma}
Next, we are going to present two distinct methods with which to compute upper bounds for the rate of convergence of expectations of functions between these two probability measures.

\subsection{Relative entropy}

We will, again, follow the presentation of this topic given in \cite{latticeexample1994}. We begin with the definition of relative entropy.
\begin{definition}[Relative entropy] Let $\lambda_1$ and $\lambda_2$ be two probability measures on a space $X$. If $\lambda_1$ is not absolutely continuous with respect to $\lambda_2$, the relative entropy $\mathcal{H}(\lambda_1 || \lambda_2) = \infty$. If $\lambda_1$ is absolutely continuous with respect to $\lambda_2$, then we have
\begin{align*}
\mathcal{H}(\lambda_1 \Vert\, \lambda_2) := \int_{X} d \lambda_1 \ln \frac{d \lambda_1}{d \lambda_2} .
\end{align*}
\end{definition}
For the paramagnet, we have the following calculation of the relative entropy.

\begin{proposition} [Upper and lower bounds for specific relative entropy] \label{upperlowerrelent} For all $m \in \operatorname{Ran}[m_N] \setminus \{ -1 , 1 \}$, $\mu \in \mathbb{R}$, and $\eta>0$, there a cutoff $N(m,\eta) \in \mathbb{N}$ such that for $N \geq N(m,\eta)$, we have
\begin{align}\label{eq:boundrelent}
 \left(\frac{1}{2}-\eta\right) \frac{\ln (N + 1)}{N} \leq \frac{\mathcal{H}(\mu^{m;N}_{\operatorname{MC}} \Vert\, \mu_{\operatorname{C}}^{\mu;N})}{N} - F(m, \mu) \leq \frac{\ln (N + 1)}{N} \,,
\end{align}
where
\begin{align}\label{eq:defFmmu}
F(m, \mu):= \ln 2 + \ln \cosh(\mu) + \mu m + \frac{1 + m}{2} \ln \left( \frac{1 + m}{2} \right) + \frac{1 - m}{2} \ln \left(  \frac{1 - m}{2} \right) \,.
\end{align}
Every 
pair $(m, \mu) \in \left( \operatorname{Ran}[m_N] \setminus \{ -1 , 1 \}\right) \times \mathbb{R}$ which satisfies $m = - \tanh \mu$ solves
$F(m, \mu) = 0$ and there are no other solutions to this equation for allowed values of $m$ and $\mu$.
\end{proposition}
\begin{proof}
We have
\begin{align*}
\frac{d \mu^{m;N}_{\operatorname{MC}}}{d \mu_{\operatorname{C}}^{\mu;N}} (\phi) = \frac{\mathbbm{1} (\phi \in \mathcal{S}_m)}{Z_{\operatorname{MC}} (m;N)} \frac{Z_{\text{C}} (\mu;N)}{e^{- \mu m N}} . 
\end{align*}
It follows that
\begin{align*}
\frac{\mathcal{H}(\mu^{m;N}_{\operatorname{MC}} \Vert\, \mu_{\operatorname{C}}^{\mu;N})}{N} = \ln 2 + \ln \cosh(\mu) + \mu m  +  \frac{\ln (N + 1)}{N} + \frac{1}{N} \ln \int_0^1 dt \ e^{N \left( \frac{1 + m}{2} \ln t + \frac{1 - m}{2} \ln (1 - t)  \right)} . 
\end{align*}
Define the function $f(t) :=  \frac{1 + m}{2} \ln t + \frac{1 - m}{2} \ln (1 - t)$. It can be shown by differentiation that the mapping $f$ is strictly concave and thus attains its unique maximum on the interval $(0,1)$. This maximum is attained at the point $t_0:= \frac{1 + m}{2} \in (0,1)$. Therefore,
\begin{align*}
\frac{1}{N} \ln \int_0^1 dt \ e^{N \left( \frac{1 + m}{2} \ln t + \frac{1 - m}{2} \ln (1 - t)  \right)} \leq f \left( \frac{1 + m}{2} \right) = \frac{1 + m}{2} \ln \left(  \frac{1 + m}{2} \right) + \frac{1 - m}{2} \ln \left(  \frac{1 - m}{2} \right) .
\end{align*} 
Recalling the definition of $F(m, \mu)$ in (\ref{eq:defFmmu}), we conclude that
\begin{align*}
\frac{\ln (N + 1)}{N} +
\frac{1}{N} \ln \int_0^1 dt \ e^{N f(t)} -f(t_0) \leq \frac{\mathcal{H}(\mu^{m;N}_{\operatorname{MC}} \Vert\, \mu_{\operatorname{C}}^{\mu;N})}{N} - F(m, \mu) \leq \frac{\ln (N + 1)}{N} \,.
\end{align*}

By computation, if $m \in (-1,1)$ and $\mu \in \mathbb{R}$ are 
such that $m = - \tanh \mu$, then $\mu = \frac{1}{2} \ln \left( \frac{1 - m}{1 + m} \right)$ and thus $F(m, \mu) = 0$.  For fixed $m$, by differentiation we find that 
$\mu \mapsto F(m,\mu)$ is a strictly convex function with a minimum at $\mu = \frac{1}{2} \ln \left( \frac{1 - m}{1 + m} \right)$. Hence, this value of $\mu$ is the only zero of the function.

To get a useful lower bound, we need to inspect the values of the function $f$ more carefully.
Suppose $0<\delta<\frac{1}{2}$ and consider $t\in [t_0 (1-\delta),t_0]$.  Then
\begin{align*}
-f''(t) = \frac{t_0}{t^2}+\frac{1-t_0}{(1-t)^2}
\leq \frac{1}{(1-\delta)^2 t_0}+\frac{1}{1-t_0} \le\frac{4}{t_0}+\frac{1}{1-t_0}\,.
\end{align*}
and since $f'(t_0)=0$, we find from Taylor's theorem that there is $K>0$ which depends only on $m$ via $t_0$ such that for the above choice of $t$, which have $|t-t_0|\le \delta t_0$,
\begin{align*}
f(t_0) - f(t) \le   K^2 \delta^2\,.
\end{align*}
Therefore, by the positivity of the integrand, we obtain
\begin{align*}
& \int_0^1 dt \ e^{N f(t)} \geq \int_{(1-\delta) t_0}^{t_0} dt \, e^{N f(t)}
\geq e^{N f(t_0) - N K \delta^2} t_0 \delta \, ,
\end{align*}
which implies 
\begin{align*}
\frac{\ln (N + 1)}{N} +
\frac{1}{N} \ln \int_0^1 dt \ e^{N f(t)} -f(t_0) \geq - K \delta^2 +\frac{1}{N} \ln (t_0(N+1)\delta)\, .
\end{align*}

Consider then arbitrary $p\in (0,\frac{1}{2})$, and set $\delta:=\frac{1}{2}(N+1)^{-\frac{1}{2}-p}$.  The lower bound becomes $-\frac{K}{4}(N+1)^{-1-2p}+\frac{1}{N} \ln (t_0(N+1)^{\frac{1}{2}-p}/2)\ge \frac{1}{N}\left(-\frac{K}{4}
+ \ln(t_0/2) + \left(\frac{1}{2}-p\right)\ln (N+1)\right)$.
Therefore, there exists $N_0$ which depends only on $m$ and $p$ such that for all $N\ge N_0$, the lower bound is
greater than $\frac{1}{N}\left(\frac{1}{2}-2p\right)\ln (N+1)$.
If $\eta<1$ we can set $p=\eta/2$, and the above estimates together thus prove (\ref{eq:boundrelent}) for all $N\ge N_0$.  The lower bound is trivial if $\eta\ge 1$.
This concludes the proof of the Proposition.
\end{proof}
\subsection{Coupling}
For the use of the coupling method, the previous thermodynamic calculations concerning the auxiliary canonical ensemble will be needed. The additional ingredient is the explicit coupling for certain observables presented in the following theorem.
\begin{theorem} \label{magcouplingthm}
Let $m \in \text{Ran}[m_N]$ and $m' \in \text{Ran}[m_N]$. We have
\begin{align*}
w_1 (\mu_{\text{MC}}^{m;N}, \mu_{\text{MC}}^{m';N};N) = |m' - m| .
\end{align*}
\end{theorem}
\begin{proof}
Denote $M = mN$ and 
$M' = m' N$. By symmetry, it suffices to prove the result in the case $M'> M$; note that if 
$M'=M$, then $\mu_{\text{MC}}^{m;N}=\mu_{\text{MC}}^{m';N}$.
First, consider any field configuration $\phi \in \mathcal{S}$ such that 
$M[\phi] = M$.
Note that $\Lambda_+ [\phi]$ corresponds to the sites on the lattice for which 
$\phi$ has a positive spin. Now, consider another field configuration such that 
$\phi' |_{\Lambda_+ [\phi]} = \phi$ and $M[\phi'] = M'$. Such a configuration 
$\phi'$ can be constructed by taking the configuration $\phi$ and flipping 
negative spins to positive spins or vice versa until we obtain the magnetization 
$M'$.

As before, define $M_+ := |\Lambda_+[\phi]|$ and $M_- := N - M_+$, and set also 
$M'_+ := |\Lambda_+[\phi']|$. 
Then  $M_+ = \frac{N + M}{2}$, $M_- = \frac{N - M}{2}$, and $M_+ < M_+'$.  
Denote $\Delta := M_+' - M_+=\frac{M'-M}{2}>0$. The number of field 
configurations with magnetization $M$ is given by $\binom{N}{M_+}$ or, 
equivalently, by $\binom{N}{M_-}$. In order to go from magnetization $M$ to 
$M'$, we must flip $\Delta$ negative sites to positive sites. The number of ways 
to do this is $\binom{M_- }{ \Delta}$. Define $\gamma : \mathcal{S} \times 
\mathcal{S} \to \mathbb{R}$ by
\begin{align*}
\gamma(\phi, \phi') := \mathbbm{1}(\phi\in\mathcal{S}_m,\, 
\phi'\in\mathcal{S}_{m'},\, \phi'|_{\Lambda_+[\phi]} = \phi|_{\Lambda_+ [\phi]}) 
\frac{1}{\binom{N }{M_-}} \frac{1}{\binom{M_- }{\Delta}} .
\end{align*}
By construction, we have
\begin{align*}
& \sum_{\phi' \in \mathcal{S}} \gamma (\phi, \phi') = 
\frac{\mathbbm{1}(\phi\in\mathcal{S}_m)}{\binom{N }{M_-}} \frac{1}{\binom{M_- }{
\Delta}} \sum_{\phi' \in \mathcal{S}_{m'}} \mathbbm{1}(\phi'|_{\Lambda_+(\phi)} 
= \phi|_{\Lambda_+ [\phi]}) 
= \frac{\mathbbm{1}(\phi\in\mathcal{S}_m)}{\binom{N }{M_-}} \frac{1}{\binom{M_- 
}{\Delta}} \binom{M_- }{\Delta} 
\\ & \quad
= \frac{\mathbbm{1}(\phi\in\mathcal{S}_m)}{\binom{N }{M_-}} 
= \frac{\mathbbm{1}(\phi\in\mathcal{S}_m)}{|\mathcal{S}_{m}|}.
\end{align*}

In the other direction, if we fix $\phi'$ with magnetization $M'$ and consider 
the number of ways to go to a field configuration $\phi$ which agrees on the 
positive lattice sites of $\phi$, then clearly we must take 
$\Delta=\frac{M'-M}{2}$ positive sites of $\phi'$ and flip them negative. The 
number of ways to do this is given by $\binom{M_+' }{\Delta}$ and we thus have
\begin{align*}
\sum_{\phi \in \mathcal{S}} \gamma (\phi, \phi') = 
\frac{\mathbbm{1}(\phi'\in\mathcal{S}_{m'})}{\binom{N }{M_-}} \frac{1}{\binom{M_- 
}{\Delta}} \sum_{\phi \in \mathcal{S}_{m}} \mathbbm{1}(\phi'|_{\Lambda_+ 
[\phi]} = \phi|_{\Lambda_+ [\phi]})
= \frac{\mathbbm{1}(\phi'\in\mathcal{S}_{m'})}{\binom{N }{M_-}} \frac{1}{\binom{M_- 
}{\Delta}} \binom{M_+' }{\Delta}\,.
\end{align*}
Now, we have the following simple binomial coefficient manipulations
\begin{align*}
\frac{\binom{M_+' }{\Delta} \binom{N }{M_+'}}{\binom{N }{M_-}\binom{M_- }{
\Delta}} = \frac{\frac{(M_+')! N!}{(M_+' - \Delta)! \Delta! (N - M_+')! 
(M_+')!}}{\frac{N! (M_-)!}{(N - M_-)! (M_-)! (M_- - \Delta)! \Delta!}} = 
\frac{(N - M_-)!}{(M_+' - \Delta)!} \frac{(M_- - \Delta)!}{(N - M_+')!} = 
\frac{(M_+)!}{(M_+)!} \frac{(M_- - M_+' + M_+)!}{(M_- + M_+ - M_+')!} = 1 .
\end{align*}
It follows that
\begin{align*}
\sum_{\phi \in \mathcal{S}} \gamma (\phi, \phi') 
= \frac{\mathbbm{1}(\phi'\in\mathcal{S}_{m'})}{\binom{N }{M_+'}} = 
\frac{\mathbbm{1}(\phi'\in\mathcal{S}_{m'})}{|\mathcal{S}_{m'}|}\, .
\end{align*}

This verifies that $\gamma$ is indeed a coupling between the fixed magnetization 
density ensembles with different magnetizations $m$ and $m'$. For such a 
coupling, by construction, we have
\begin{align*}
\gamma (\phi, \phi')\, || \phi - \phi' ||_1 = 2 \Delta \gamma (\phi, \phi') = 
(M' - M) \gamma(\phi, \phi'),
\end{align*}
from which it follows that
\begin{align*}
w_1 (\mu_{\text{MC}}^{m;N}, \mu_{\text{MC}}^{m';N};N) \leq \frac{M' - M}{N} = m' 
- m .  
\end{align*}
On the other hand, if $\eta$ is any other coupling of $\mu_{\text{MC}}^{m;N}$ 
and $\mu_{\text{MC}}^{m';N}$, we also have
\begin{align*}
\frac{1}{N} \int_{\mathcal{S} \times \mathcal{S}} \eta (d \phi, d \phi') \sum_{x 
\in \Lambda} |\phi(x) - \phi'(x)| \geq \frac{1}{N} \left| \left< \sum_{x \in 
\Lambda} \phi(x) - \sum_{x \in \Lambda} \phi'(x) \right>_{\eta} \right| = m' - m 
.
\end{align*}
This implies that the coupling $\gamma$ is an optimal coupling, and, we have
\begin{align*}
w_1 (\mu_{\text{MC}}^{m;N}, \mu_{\text{MC}}^{m';N};N) = m' - m .
\end{align*}
This completes the proof assuming $M'>M$, and hence by symmetry, also the proof 
of the Theorem.
\end{proof}

\subsection{Convergence of local observables}
In this subsection, we will present two distinct proofs of the convergence of local observables based on the two previously introduced objects. The first proof will be based on utilizing the Pinsker inequality and its relationship with relative entropy. This type of argument can be found in \cite{Grokinsky2008}. The argument uses the information divergence related methods from \cite{csiszar1975}
in which a statement concerning the relationship between weak convergence and relative entropy is also given.

We must also emphasize that whatever we refer to here as the relative entropy method is precisely the collection of arguments and theorems that will be presented shortly. We are not stating that the classical inequalities could not be, for instance, strengthened or leveraged with other theorems in order to produce better results. For an example of this sort of work, we refer to \cite{Chleboun2013}.  

\subsubsection{Relative entropy method}
Again, we will follow the example of \cite{latticeexample1994}. We have the following theorem.
\begin{theorem} \label{relentinequality}
Let $m \in \text{Ran}[m_N] \setminus \{ -1 , 1 \}$ and $\mu \in \mathbb{R}$. Let $I \subset \Lambda$ and suppose $f : 
\{ -1,1 \}^{|I|} \to \mathbb{R}$. It follows that
\begin{align*}
\frac{\left| \left< f \circ P_I \right>_{\text{MC}}^{m;N} - \left< f \circ P_I 
\right>_{\text{C}}^{\mu;N} \right|^2}{2}  \leq  |I| \left( \max_{\phi \in \{ -1, 1\}^{|I|}} |f(\phi)| \right)^2  \frac{\mathcal{H}(\mu^{m;N}_{\operatorname{MC}} \Vert\, \mu_{\operatorname{C}}^{\mu;N})}{N}
\end{align*}
\end{theorem}
\begin{proof} The following proof can be considered a sketch of the standard argument presented for the same model in \cite{latticeexample1994}.

Since $f$ is a bounded function, and the measures are absolutely continuous with respect to each other, we have
\begin{align*}
\left| \left< f \circ P_I \right>_{\text{MC}}^{m;N} - \left< f \circ P_I 
\right>_{\text{C}}^{\mu;N} \right| \leq  \left( \max_{\phi \in \{ -1, 1\}^{|I|}} |f(\phi)|  \right) || \mu_{\operatorname{MC}}^{m;N} |_{I} - \mu_{\operatorname{C}}^{\mu;N} |_{I}||_{\operatorname{TV}} .
\end{align*}
By Pinsker's inequality, we have
\begin{align*}
\frac{|| \mu_{\operatorname{MC}}^{m;N} |_{I} - \mu_{\operatorname{C}}^{\mu;N} |_{I}||^2_{\operatorname{TV}}}{2} \leq \mathcal{H}(\mu_{\operatorname{MC}}^{m;N} |_{I} || \mu_{\operatorname{C}}^{\mu;N} |_{I}) .
\end{align*}
Now, let $N \geq |I|$ and let $I_k \subset \Lambda$ be disjoint copies of size $|I|$ on the lattice such that $\Lambda \subset \bigcup_{k=1}^{K} I_k$ and $\bigcup_{k=1}^{K - 1} I_k \subset \Lambda$. It follows that $(K - 1) |I| \leq |\Lambda| \leq K |I|$ which implies that $\frac{1}{K} \leq \frac{|I|}{N}$. Now, by utilizing the fact that $\mu_{\operatorname{C}}^{\mu;N}$ is a product measure, and that both measures are permutation invariant, it follows that
\begin{align*}
\mathcal{H}(\mu_{\operatorname{MC}}^{m;N} |_{I} \Vert \mu_{\operatorname{C}}^{\mu;N} |_{I}) \leq |I| \frac{\mathcal{H}(\mu^{m;N}_{\operatorname{MC}} \Vert\, \mu_{\operatorname{C}}^{\mu;N})}{N}.
\end{align*}
The statement follows by combining these calculations.
\end{proof}

We give the full convergence rate in the following corollary.   Here, and in the following, we employ the standard rigorous definition of the ``$O$''-notation: given $g(N)\ge 0$ for $N\in \N$,
``$X(N)=O(g(N))$'' refers to the limit $N\to \infty$, i.e., it means there is an
$N$-independent constant $C$ and some $N_0\in \N$ such that $|X(N)|\le C g(N)$ for all $N\ge N_0$.  However, 
in these results, the constant $C$ is allowed to depend on possible other parameters of the setup: for example, no uniformity of $C$ in the parameters $m$, $|I|$ or $\norm{f}_\infty$ is claimed below.

\begin{corollary}\label{th:relentcorr}
Let $m \in \text{Ran}[m_N] \setminus\{ -1,1\}$ and $\mu = 
\tanh^{-1} (- m)$. Let $I \subset \Lambda$ be a fixed size index set, and let $f : 
\{ -1,1 \}^{|I|} \to \mathbb{R}$. It follows that
\begin{align*}
\left| \left< f \circ P_I \right>_{\text{MC}}^{m;N} - \left< f \circ P_I 
\right>_{\text{C}}^{\mu;N} \right| &\leq  \sqrt{2 |I|} \left( \max_{\phi \in \{ -1, 1\}^{|I|}} |f(\phi)| \right) \sqrt{\frac{\mathcal{H}(\mu^{m;N}_{\operatorname{MC}} \Vert\, \mu_{\operatorname{C}}^{\mu;N})}{N}} \\ &\leq \sqrt{2 |I|} \left( \max_{\phi \in \{ -1, 1\}^{|I|}} |f(\phi)| \right)\sqrt{\frac{\log(N + 1)}{N}} .
\end{align*}
By applying the relative entropy method, we have
\begin{align*}
\left< f \circ P_I \right>_{\text{MC}}^{m;N} = \left< f \circ P_I 
\right>_{\text{C}}^{\mu;N} + O \left( \sqrt{\log (N)} N^{- \frac{1}{2}} \right) .
\end{align*}
The factor of $\sqrt{\log(N)}$ cannot be removed by using this specific inequality since 
if we choose $\eta=\frac{1}{4}$, 
there exists a cutoff $N(m) \in \mathbb{N}$ such that for $N \geq N(m)$, we have
\begin{align*}
  \frac{\ln (N + 1)}{4 N} \leq \frac{\mathcal{H}(\mu^{m;N}_{\operatorname{MC}} \Vert\, \mu_{\operatorname{C}}^{\mu;N})}{N} .
\end{align*}
\end{corollary}
\begin{proof}
The corollary follows directly by combining the contents and bounds from \cref{upperlowerrelent} and \cref{relentinequality}.
\end{proof}

\subsubsection{Coupling method}
The main theorems formulated in the coupling section concern $1$-Lipschitz 
functions with respect to some norm $|| \cdot ||_p$. 
Since the domain set is finite, all functions $f : \{ -1 , 1\}^{|I|} \to 
\mathbb{R}$ are automatically Lipschitz functions with respect to all of these 
norms. 
The choice of using $p=1$ norm below is partially a matter of convenience, due 
to equivalence of the finite set $p$-norms, but one should be careful in the 
application 
of the result if the size of the set $I$ is allowed to become unbounded as 
$N\to\infty$.  We recall from Sec.~\ref{sec:Wassandoptim} that the optimal Lipschitz constant does depend on the choice of 
norm, and it will affect the overall constant in the bounds, unless scaled to one, as we require here.

We can now state the full convergence theorem. We continue to use the 
notations introduced before Corollary \ref{th:relentcorr}, i.e., 
``$A(N)= B(N)+ O(N^{-\frac{1}{2}})$'' here means that there exists
$N_0\in \N$ and $C\ge 0$ 
such that $|A(N)-B(N)|\le C N^{-\frac{1}{2}}$ for all $N\ge N_0$.

\begin{theorem} \label{magequiv} Let $m \in \text{Ran}[m_N] \setminus\{ -1,1\}$ and $\mu = 
\tanh^{-1} (- m)$. Let $I \subset \Lambda$ be a fixed size index set, and let $f : 
\{ -1,1 \}^{|I|} \to \mathbb{R}$ be a bounded $1$-Lipschitz function with 
respect to the $|| \cdot ||_1$ norm. It follows that
\begin{align*}
\left< f \circ P_I \right>_{\text{MC}}^{m;N}  = \left< f \circ P_I 
\right>_{\text{C}}^{\mu;N} + O(N^{- \frac{1}{2}}) . 
\end{align*}
\begin{proof}
The result follows by applying the free energy method presented in 
\Cref{freeenergylip}, along with the $w_1$ 
fluctuation distance bound presented in \Cref{magcouplingthm}, and the equations 
in \Cref{partitionmag}. 
\end{proof}
\end{theorem}
For the auxiliary microcanonical ensemble with fixed magnetization density, the $w_1$ 
choice of cost function is natural since the $|| \cdot ||_1$-norm satisfies
\begin{align*}
|M[\phi] - M[\phi']| \leq || \phi - \phi' ||_1 .
\end{align*}

\subsection{Relationship to the Curie--Weiss model}\label{sec:reltoCW}

Let us first recall the Curie--Weiss Hamiltonian.
\begin{definition}[Curie--Weiss Hamiltonian]
Let $\Lambda$ be a finite lattice 
with size $N := |\Lambda|$. We define the space of spin configurations 
$\mathcal{S} := \{ -1 , 1 \}^\Lambda$. Let $J > 0$ and $h \in \mathbb{R}$. 
Define the Hamiltonian $H : \mathcal{S} \to \mathbb{R}$ by
\begin{align*}
H[\phi] := - \frac{J}{2 N} \sum_{x,y \in \Lambda} \phi(x) \phi(y) - h \sum_{x 
\in \Lambda} \phi(x) .
\end{align*}
The energy density $\varepsilon_N : \mathcal{S} \to  
\mathbb{R}$ is defined by 
$\varepsilon_N [\phi] := \frac{H[\phi]}{N}$. 
\end{definition}
Note that the Hamiltonian can be written in terms of the magnetization by
\begin{align*}
H[\phi] = - \frac{J}{2N} M[\phi]^2 - h M[\phi] = - \frac{J}{2 N} \left( M[\phi] 
+ \frac{h N}{J} \right)^2 + \frac{h^2 N}{2 J}.
\end{align*}
This relation leads to a simplification when studying the microcanonical ensemble of the Curie--Weiss model, defined as follows.  
\begin{definition}[Fixed energy density/Microcanonical ensemble] Let 
$\varepsilon \in \mathbb{R}$ be such that $\varepsilon  \in 
\text{Ran}[\varepsilon_N]$. Define the set $\mathcal{S}_\varepsilon := \{ \phi 
\in \mathcal{S} : \varepsilon_N[\phi] = \varepsilon \}$. The microcanonical 
ensemble with energy density $\varepsilon$ is defined via its action on
functions
$f : \mathcal{S} \to \mathbb{R}$ by
\begin{align*}
\left< f \right>_{\text{MC}}^{\varepsilon;N} := 
\frac{1}{|\mathcal{S}_\varepsilon|} \sum_{\phi \in \mathcal{S}_\varepsilon} f 
(\phi) .
\end{align*}
\end{definition}
We will always use the lower case letter $m$ to specify fixed magnetization densities introduced in Definition \ref{th:defparammicro},
and $\varepsilon$ for fixed energy densities so that there is no ambiguity.

In some sense, the fixed energy ensemble for some values of $J$ and $h$ is not 
necessarily fundamental as it can be represented as a convex combination of 
fixed magnetizations.
The energy density can be written in terms of the magnetization density as
\begin{align*}
\varepsilon_N[\phi] = - \frac{J}{2} \left( m_N[\phi] + \frac{h}{J} \right)^2 + 
\frac{h^2}{2 J} \iff m_{N,\pm}[\phi] = - \frac{h}{J} \pm \sqrt{\frac{h^2}{J^2} - 
\frac{2 \varepsilon_N[\phi]}{J}}\ ,
\end{align*}
from which it is clear that for some values of $h$ and $J$ there are multiple 
magnetization densities which give the same energy density. The following lemma 
makes the previous statements more quantitative.
\begin{lemma} \label{energysplitlem} Let $\varepsilon \in \mathbb{R}$ be such 
that $\varepsilon \in \text{Ran}[\varepsilon_N]$ and define $m_\pm = - 
\frac{h}{J} \pm \sqrt{\frac{h^2}{J^2} - \frac{2 \varepsilon}{J}}$. We have
\begin{align}\label{eq:fmpldecomp}
\left< f \right>_\text{MC}^{\varepsilon;N} = 
\frac{|\mathcal{S}_{m_+}|}{|\mathcal{S}_{m_+}| + |\mathcal{S}_{m_-}|} \left< f 
\right>_{\text{MC}}^{m_+;N} + \frac{|\mathcal{S}_{m_-}|}{|\mathcal{S}_{m_+}| + 
|\mathcal{S}_{m_-}|} \left< f \right>_{\text{MC}}^{m_-;N},
\end{align}
with the convention that $\left< f \right>_{\text{MC}}^{m;N} = 0$ and 
$|\mathcal{S}_m| = 0$ if $m \not \in \text{Ran}[m_M]$.
\end{lemma}
\begin{proof} 
Since $\varepsilon \in \text{Ran}[\varepsilon_N]$, it follows that $\vep\le 
\frac{h^2}{2 J}$.  If $\vep= \frac{h^2}{2 J}$, we have $m_+=m_-$, 
$\mathcal{S}_{m} = \mathcal{S}_\vep$ for $m=m_\pm$,
and thus $\left< f \right>_\text{MC}^{\varepsilon;N} = \frac{1}{2} \left< f 
\right>_{\text{MC}}^{m_+;N} + \frac{1}{2}\left< f \right>_{\text{MC}}^{m_-;N}$.  
Hence, (\ref{eq:fmpldecomp}) holds in this case.

We may thus assume that $\vep< \frac{h^2}{2 J}$, when $m_-< m_+$.
We have $\mathcal{S}_\varepsilon = \mathcal{S}_{m_+} \cup \mathcal{S}_{m_-}$ and 
$\mathcal{S}_{m_+} \cap \mathcal{S}_{m_-} = \emptyset$.  Thus,
$|\mathcal{S}_{\varepsilon}| = |\mathcal{S}_{m_-}| + |\mathcal{S}_{m_+}|$ and
\begin{align*}
\frac{1}{|\mathcal{S}_\varepsilon|} \sum_{\phi \in \mathcal{S}_\varepsilon} f 
(\phi) = \frac{1}{|\mathcal{S}_{m_-}| + |\mathcal{S}_{m_+}|} \left( 
|\mathcal{S}_{m_+}| \frac{1}{|\mathcal{S}_{m_+}|}\sum_{\phi \in 
\mathcal{S}_{m_+}} f (\phi) + |\mathcal{S}_{m_-}| 
\frac{1}{|\mathcal{S}_{m_-}|}\sum_{\phi \in \mathcal{S}_{m_-}} f (\phi)\right),
\end{align*}
from which the statement follows.
\end{proof}
From the previous lemma, it is apparent that if one has knowledge of the auxiliary microcanonical partition function and local observables of the fixed magnetization ensemble, then, in principle, one has a full description of the weak convergence properties of corresponding fixed energy density ensemble. In the later sections, we will give an example of this exact kind of analysis with the mean-field spherical  model.

\subsection{Final remarks and comparison of methods}

From the results in the previous subsections, we can see that the coupling method generates a strictly better convergence rate by removing the factor of $\sqrt{\log(N)}$ which is an irremovable part of the magnitude of specific relative entropy in this case.
For the relative entropy method, the primary object of interest is the calculation of the specific relative entropy. As can be seen, the selection of the magnetization density $m$ and the parameter $\mu$ comes down to solving an equation which relates these parameters. Once the pair has been realized from this equation, we obtain the desired upper bound. In our calculation, the partition functions of both the auxiliary microcanonical ensemble and auxiliary canonical ensemble needed to be estimated sufficiently accurately (other more refined approaches for using relative entropy may be found in \cite{LJP1995}).

For the coupling method, one needed to come up with the coupling of the auxiliary microcanonical ensembles resulting in 
Theorem \ref{magcouplingthm}.  In addition, one uses 
the standard ``thermodynamic'' relations for the auxiliary canonical ensemble, given in Lemma \ref{partitionmag}. Because the auxiliary canonical ensemble was in product form, the calculations were particularly simple. 

From these observations, the main difference between the methods concerns the treatment of the auxiliary microcanonical ensemble. 
In the above direct relative entropy computation, we  need to calculate the partition function of the auxiliary microcanonical ensemble and the auxiliary canonical ensemble, but the ``thermodynamic'' relations of the auxiliary
canonical ensemble do not seem important. For the coupling method, the auxiliary microcanonical partition function does not play the same role, and the coupling is the most important object along with the ``thermodynamic'' relations from the auxiliary canonical ensemble.

\section{Mean-field spherical model aka Continuum Curie--Weiss model}\label{sec:ContCW}

For this model, we will need to clarify the goals and priority of the limiting measures for the local convergence result. Our main goal is to analyse explicitly the probability measure associated with the microcanonical ensemble. To that end, there will be some local convergence results in which thermodynamic equivalence might hold between two ensembles, but we will opt for a simpler auxiliary ensemble measure as the approximating measure.

In particular, for this model, we will see that the simplest limit measure to consider will be either a Gaussian measure or a convex combination of Gaussian measures. We will also prove some local convergence results where the limiting measure is not a product measure.
 
In the second model considered here, the ``spin-field'' $\phi$ is allowed to take all real values otherwise being similar to the discrete Curie--Weiss model.
\begin{definition}[Continuum Curie--Weiss Hamiltonian]
Let $\Lambda$ be a finite lattice 
with size $N := |\Lambda|$. We define the space of field configurations 
$\mathcal{S} := \mathbb{R}^\Lambda$. Let $J > 0$ and $h \in \mathbb{R}$. Define 
the Hamiltonian $H : \mathcal{S} \to \mathbb{R}$ and the particle number $N : 
\mathcal{S} \to \mathbb{R}$ by
\begin{align*}
H[\phi] := - \frac{J}{2 N} \sum_{x,y \in \Lambda} \phi(x) \phi(y) - h \sum_{x 
\in \Lambda} \phi(x), \quad N[\phi] := \sum_{x \in \Lambda} \phi(x)^2 .
\end{align*}
We also define the magnetization $M : \mathcal{S} \to \mathbb{R}$ by
\begin{align*}
M[\phi] := \sum_{x \in \Lambda} \phi(x) .
\end{align*}
Furthermore, we define the energy density $\varepsilon : \mathcal{S} \to 
\mathbb{R}$ and magnetiztion density $m : \mathcal{S} \to \mathbb{R}$ by 
$\varepsilon [\phi] := \frac{H[\phi]}{N}$ and $m [\phi] := \frac{M[\phi]}{N}$. 
\end{definition}
As in the discrete case, the Hamiltonian can be written in terms of the magnetization,
\begin{align*}
H[\phi] = - \frac{J}{2N} M[\phi]^2 - h M[\phi] = - \frac{J}{2 N} \left( M[\phi] 
+ \frac{h N}{J} \right)^2 + \frac{h^2 N}{2 J}.
\end{align*}
In this model, the particle number function $N[\cdot]$ is much more relevant than 
in the discrete case. For this Hamiltonian, we will need to consider probability 
measures described by products of delta functions.  To properly resolve them, we begin with an observation concerning a matrix relevant to the definitions of the 
ensembles.  In the following, we employ the notation $M_N (\mathbb{R})$ for the collection of real $N\times N$ matrices.

\begin{lemma} \label{changeofvar} 
Define $M \in M_N (\mathbb{R})$ by $M_{ij} := 
1$ for all $i \in [N]$ and $j \in [N]$.  There exists an 
orthogonal matrix $U \in M_N(\mathbb{R})$ which diagonalizes $M$ such that for 
$x \in \mathbb{R}^N$, we have
\begin{align*}
(U x)_1 = \frac{1}{\sqrt{N}} \sum_{i=1}^N x_i .
\end{align*}
\end{lemma}
\begin{proof}
A simple analysis shows that $M$ has an eigenvalue $N$ with no degeneracy, and 
an eigenvalue $0$ with $N-1$ fold degeneracy.   Collecting the eigenvalues into a diagonal matrix results in $D \in M_N (\mathbb{R})$ for which  
$D_{11} := N$ and $D_{ij} := 0$ for all other $i \in [N]$ and $j \in [N]$. 
Since $M$ is a real symmetric matrix, then there exists an orthogonal matrix $Q\in M_N (\mathbb{R})$ 
such that 
\begin{align*}
Q^T D Q = M .
\end{align*}
Writing out the above matrix multiplication componentwise explicitly, we find for all $i,j$
\begin{align*}
 N Q_{1i} Q_{1j} = 1 .
\end{align*}
In particular, then $|Q_{1i}| = \frac{1}{\sqrt{N}}$ for all $i \in [N]$, and thus for each $i$ there is $\sigma_i \in \{\pm 1\}$ such that $Q_{1i} = \sigma_i \frac{1}{\sqrt{N}}$. Using a proof by contradiction, 
one can see that, in fact, the elements $Q_{1i}$ must either all be negative or 
all be positive. Now, define $U \in M_N(\mathbb{R})$ by $U := - Q$ if the 
elements $Q_{1i}$ are all negative, and $U  := Q$ if the elements $Q_{1i}$ are 
all positive. It follows that $U$ is an orthogonal matrix and, by definition, we 
have
\begin{align*}
(Ux)_1 := \sum_{j=1}^N U_{1j} x_j = \frac{1}{\sqrt{N}} \sum_{i=1}^N x_i .
\end{align*}
This completes the proof of the Lemma.
\end{proof}

Next, we will give two examples of how to apply $\delta$-function calculation rules  to resolve the ones relevant to the Curie--Weiss system, for both
microcanonical and canonical ensembles.  It is possible to prove the validity of these manipulations under the assumptions made in the Examples, for instance, following the discussion in Appendix A of \cite{Lukk18}.

\begin{example} \label{magnetex}
Let $U : \mathcal{S} \to \mathbb{R} \times \mathbb{R}^{N - 1}$ 
be an orthogonal matrix satisfying \Cref{changeofvar}. Suppose $\rho > 0$ and $m 
\in (-\sqrt{\rho},\sqrt{\rho})$. For bounded $1$-Lipschitz functions $f : 
\mathcal{S} \to \mathbb{R}$, we have 
\begin{align*}
&\int_{\mathcal{S}} d \phi \ \delta \left( \sum_{x \in \Lambda} \phi(x) - m N 
\right) \delta \left( \sum_{x \in \Lambda} \phi(x)^2 - \rho N \right) f(\phi) \\ 
&=\int_{\mathbb{R}} d z \ \int_{\mathbb{R}^{N - 1}} d \psi \  
\delta (\sqrt{N} z - m N) \delta(z^2 + || \psi ||^2 - \rho N) \left( f \circ 
U^{-1} \right) (z , \psi) \\
&= \frac{(N(\rho - m^2))^{\frac{N - 3}{2}}}{2 \sqrt{N}} \int_{\mathbb{S}^{N - 
2}} d \Omega \ \left(f \circ U^{-1} \right) (m \sqrt{N}, \sqrt{N(\rho - m^2)} 
\Omega) .
\end{align*} 
Here we have first made a change of variables to $(z,\psi) = U\phi$ and then used spherical coordinates system to integrate out the resulting $\delta$-functions.
Since the left hand side does not depend on the choice of the matrix $U$, all choices must result in the same value for the integral on the right hand side.
\end{example}

\begin{example} \label{energyex} Let $U : \mathcal{S} \to \mathbb{R} \times 
\mathbb{R}^{N - 1}$ be 
an orthogonal matrix satisfying \Cref{changeofvar}. Fix 
$h \in \mathbb{R}$ and $\rho > 0$, and suppose $\varepsilon \in \mathbb{R}$ 
satisfies $\vep<\frac{h^2}{2 J}$.  Define then
\begin{align*}
m_+ := - \frac{h}{J} + \sqrt{\frac{h^2}{J^2} - \frac{2 \varepsilon}{J}} , \quad m_- 
:=  - \frac{h}{J} - \sqrt{\frac{h^2}{J^2} - \frac{2 \varepsilon}{J}} .
\end{align*} 
Then, $m_-, m_+$ are distinct real numbers, and we assume furthermore that $m_+^2 < \rho$ and $m_-^2 <\rho$. 

Then for bounded 
$1$-Lipschitz functions $f : \mathcal{S} \to \mathbb{R}$, we may proceed as in the previous Example to conclude that
\begin{align*}
&\int_{\mathcal{S}} d \phi \ \delta \left( - \frac{J}{2 N} \sum_{x,y \in 
\Lambda} \phi(x) \phi(y) - h \sum_{x \in \Lambda} \phi(x) - \varepsilon N 
\right) \delta \left( \sum_{x \in \Lambda} \phi(x)^2 - \rho N \right) f(\phi) \\ 
&\quad = \int_{\mathbb{R}} d z \ \int_{\mathbb{R}^{N - 1}} d \psi \  \delta \left(- 
\frac{J}{2} z^2- h \sqrt{N} z - \varepsilon N\right) \delta(z^2 + || \psi ||^2 - 
\rho N) \left( f \circ U^{-1} \right) (z , \psi) \\
&\quad = \frac{1}{J \sqrt{\frac{h^2}{J^2} - \frac{2 \varepsilon}{J}} \sqrt{N}} 
\int_{\mathbb{R}} d z \ \int_{\mathbb{R}^{N - 1}} d \psi \  \delta \left(z - m_+ 
 \sqrt{N}\right) \delta(z^2 + || \psi ||^2 - \rho N) \left( f \circ U^{-1} 
\right) (z , \psi)  \\
&\qquad + \frac{1}{J \sqrt{\frac{h^2}{J^2} - \frac{2 \varepsilon}{J}} \sqrt{N}} 
\int_{\mathbb{R}} d z \ \int_{\mathbb{R}^{N - 1}} d \psi \  \delta \left(z - m_- 
 \sqrt{N}\right) \delta(z^2 + || \psi ||^2 - \rho N) \left( f \circ U^{-1} 
\right) (z , \psi)  \\
&\quad = \frac{1}{J \sqrt{\frac{h^2}{J^2} - \frac{2 \varepsilon}{J}} \sqrt{N}} 
\frac{(N(\rho - m_+^2))^{\frac{N - 3}{2}}}{2 \sqrt{N}} \int_{\mathbb{S}^{N - 2}} 
d \Omega \ \left(f \circ U^{-1} \right) (m_+ \sqrt{N}, \sqrt{N(\rho - m_+^2)} 
\Omega)  \\
&\qquad + \frac{1}{J \sqrt{\frac{h^2}{J^2} - \frac{2 \varepsilon}{J}} \sqrt{N}} 
\frac{(N(\rho - m_-^2))^{\frac{N - 3}{2}}}{2 \sqrt{N}} \int_{\mathbb{S}^{N - 2}} 
d \Omega \ \left(f \circ U^{-1} \right) (m_- \sqrt{N}, \sqrt{N(\rho - m_-^2)} 
\Omega)    .
\end{align*} 
\end{example}
We will 
utilize these forms for more explicit definitions of the ensembles and in the proof concerning the boundedness of moments of the 
microcanonical ensembles.

\subsection{Microcanonical analysis}

From here on, whenever the mapping $U$ is present, we are always referring 
to the mapping $U$ defined by a matrix satisfying \Cref{changeofvar}.  We fix the choice of this matrix in the following.
We begin with definitions of the two
ensembles, related to fixed magnetization and to fixed energy.
\begin{definition}[Fixed magnetization density and particle 
density/auxiliary microcanonical ensemble] Let $\rho > 0$ and $m \in (- 
\sqrt{\rho},\sqrt{\rho})$. 
The auxiliary microcanonical ensemble with particle density 
$\rho$ and magnetization density $m$ is defined via its action on bounded 
$1$-Lipschitz functions $f : \mathcal{S} \to \mathbb{R}$ by
\begin{align*}
\left< f \right>_{\text{MC}}^{m, \rho;N} := \frac{1}{|\mathbb{S}^{N - 2}|} 
\int_{\mathbb{S}^{N - 2}} d \Omega \ \left(f \circ U^{-1} \right) (m \sqrt{N}, 
\sqrt{N(\rho - m^2)} \Omega) .
\end{align*}
Furthermore, we define the specific auxiliary microcanonical partition function by
\begin{align*}
Z_\text{MC} (m, \rho;N) := (\rho - m^2)^{\frac{N - 3}{2}} ,
\end{align*}
and the specific auxiliary microcanonical entropy by
\begin{align*}
s_\text{MC}(m, \rho) := \frac{1}{2} \ln (\rho - m^2) .
\end{align*}
\end{definition}
\begin{definition}[Fixed energy density and particle density/microcanonical 
ensemble]\label{th:defcontfixedE}
Let $\rho > 0$, and $\varepsilon$, $m_+$, and $m_-$ be as in 
\Cref{energyex}, in particular, assume $\vep<\frac{h^2}{2 J}$ and $m_-^2, m_+^2 < \rho$.
The microcanonical ensemble with energy density $\varepsilon$ 
and particle density $\rho > 0$ is then defined via its action on bounded 
$1$-Lipschitz functions $f : \mathcal{S} \to \mathbb{R}$ by
\begin{align*}
& \left< f \right>_\text{MC}^{\varepsilon, \rho;N} := \frac{Z_\text{MC} (m_+, 
\rho;N)}{Z_\text{MC} (m_+, \rho;N) + Z_\text{MC} (m_-, \rho;N)} \left< f 
\right>_{\text{MC}}^{m_+, \rho;N}   \\
& \qquad + \frac{Z_\text{MC} (m_-, \rho;N)}{Z_\text{MC} 
(m_+, \rho;N) + Z_\text{MC} (m_-, \rho;N)} \left< f \right>_{\text{MC}}^{m_-,\rho;N} .
\end{align*}

If $\vep<\frac{h^2}{2 J}$ but 
$\min(m_-^2,m_+^2)<\rho\le \max(m_-^2,m_+^2)$, 
we set $\left< f \right>_\text{MC}^{\varepsilon, \rho;N} :=  \left< f \right>_{\text{MC}}^{m, \rho;N}$ with 
$m:= m_+$ if $|m_+|< |m_-|$, and $m := m_-$, otherwise.

If $\vep=\frac{h^2}{2 J}$ and $\rho>\frac{h^2}{J^2}$, we set $\left< f \right>_\text{MC}^{\varepsilon, \rho;N} :=  \left< f \right>_{\text{MC}}^{m, \rho;N}$ with 
$m:= -\frac{h}{J}$. 
\end{definition}

One can indeed verify by using the calculations in Examples \ref{magnetex} and 
\ref{energyex} that these measures correspond to 
$\delta$-function definitions, resolved in the manner used in the Examples.  The second definition serves as an explanation of the choice of multiplicative constants 
in the definition of the microcanonical partition function and specific entropy.
In the degenerate case $\vep=\frac{h^2}{2 J}$,
we would have above $m_+=m_-=-\frac{h}{J}=m$, 
and in this case the $\delta$-function definition in Example \ref{energyex}
does not really make sense since it would contain a singular term $\delta\left((z-m\sqrt{N})^2\right)$. Following these observations, we 
define  the fixed energy microcanonical ensemble via the 
corresponding fixed magnetization ensemble.

Note that in addition to values of $(\vep,\rho)$ for which there are no solutions to the constraints, we have also left undefined the degenerate energy ensembles for which 
$\vep\le \frac{h^2}{2 J}$ but $\min(m_-^2,m_+^2)=\rho$, as well as the degenerate magnetization ensembles with $m^2=\rho$.  In theses cases, the dimensionality of the solution manifold does not increase with $N$ since all solutions have
$\psi=0$.  As such, the resulting degenerate ensemble does not have standard thermodynamic behaviour.

We begin by estimating the fluctuation distance of two 
fixed magnetization ensembles by constructing a suitable transport map between them.
\begin{theorem} \label{magcouplingthm2} Let $\rho > 0$ and $m, m' \in 
(-\sqrt{\rho},\sqrt{\rho})$. Consider some $N\in \N$.  We have
\begin{align*}
w_2 (\mu_{\text{MC}}^{m, \rho;N}, \mu_{\text{MC}}^{m', \rho;N};N) \leq \left( 1 
+ \frac{2}{\sqrt{1-m^2/\rho}} \right) |m - m'|.
\end{align*}
\end{theorem}
\begin{proof} Define $T : \mathbb{R} \times \mathbb{R}^{N - 1} \to \mathbb{R} 
\times \mathbb{R}^{N - 1}$ by
\begin{align*}
T(z, \psi) := \left( m' \sqrt{N}, \ \sqrt{N(\rho - m'^2)} \frac{\psi}{|| \psi 
||_2}\right) .
\end{align*}
Note that then for any $\Omega \in \mathbb{S}^{N - 2}$ we have $T(m \sqrt{N}, 
\sqrt{N(\rho - m^2)} \Omega)=(m' \sqrt{N}, 
\sqrt{N(\rho - (m')^2)} \Omega)$.
In order for this mapping to act on the correct coordinate space, we define 
$T' : \mathbb{R}^\Lambda \to \mathbb{R}^\Lambda$ by $T' := U^{-1} \circ T \circ 
U$.  Then for any observable $f$ we obtain directly from the definitions a relation
\begin{align*}
\left< f \circ T' \right>_{\text{MC}}^{m, \rho;N} = \left< f 
\right>_{\text{MC}}^{m', \rho;N}  .
\end{align*}
Therefore, $T'$ is a transport map from the measure $\mu_{\text{MC}}^{m, \rho;N}$ to $\mu_{\text{MC}}^{m', \rho;N}$.  Let $\gamma$ denote the associated
coupling as defined in Sect.\ \ref{sec:couplandtransport}.  This yields an estimate
\begin{align*}
& w_2 (\mu_{\text{MC}}^{m, \rho;N}, \mu_{\text{MC}}^{m', \rho;N};N)^2 
\le \int \gamma(\rmd \phi,\rmd \psi) \frac{1}{N} \norm{\phi-\psi}_2^2
= \int \mu_{\text{MC}}^{m, \rho;N}(\rmd \phi) \frac{1}{N} \norm{\phi-T'\phi}_2^2
\\ & \quad
=
\frac{1}{N} \frac{1}{|\mathbb{S}^{N - 2}|} \int_{\mathbb{S}^{N - 2}} d \Omega \ 
\left| \left| U^{-1} (m \sqrt{N}, \sqrt{N(\rho - m^2)} \Omega)  - U^{-1} (m' 
\sqrt{N}, \sqrt{N(\rho - m'^2)} \Omega)\right| \right|^2 \\
& \quad = \frac{1}{N}\frac{1}{|\mathbb{S}^{N - 2}|} \int_{\mathbb{S}^{N - 2}} d \Omega 
\ \left| \left| (m \sqrt{N}, \sqrt{N(\rho - m^2)} \Omega)  -  (m' \sqrt{N}, 
\sqrt{N(\rho - m'^2)} \Omega)\right| \right|^2 \\
& \quad = (m - m')^2 + \left( \sqrt{\rho - m^2} - \sqrt{\rho - m'^2}\right)^2 \leq 
\left( 1 + \frac{4\rho}{\rho - m^2} \right) (m - m')^2 .
\end{align*}
Since $1+x^2 \le (1+x)^2$ for $x\ge 0$, we obtain the stated bound after taking a square root.
\end{proof}
In the previous theorem, the fluctuation distance is seemingly bounded asymmetrically with respect to the magnetization densities $m$ and $m'$. 
By symmetry, the bound holds for either choice, and thus a symmetric bound
can also be straightforwardly derived.
The reason for the asymmetric choice is that while using the fluctuation distance, we will always consider one of the magnetization densities to be fixed.

To study the fixed energy ensembles, we begin with a Lemma
which implies that, for $h \not = 0$, one 
of the fixed magnetization measures dominates in the fixed energy ensemble.
\begin{lemma} \label{maxzlemma} Consider some $N\in \N$. If $h \not = 0$, then 
\begin{align*}
\max \{ Z_\text{MC} (m_+, \rho;N), Z_\text{MC} (m_-, \rho;N) \} = \left( \rho - 
\left( \frac{|h|}{J} - \sqrt{\frac{h^2}{J^2} - \frac{2 
\varepsilon}{J}}\right)^2\right)^{\frac{N - 3}{2}} .
\end{align*}
\end{lemma}
\begin{proof}
If we consider the mapping $m \mapsto Z_\text{MC} (m, \rho;N)$, then it is clear 
that $Z_\text{MC} (m, \rho;N) \geq Z_\text{MC} (m', \rho;N)$ for all $|m| \leq 
|m'|<\sqrt{\rho}$. Now, note that
\begin{align*}
m_{\pm}^2 = \frac{h^2}{J^2} + \left( \frac{h^2}{J^2} - \frac{2 \varepsilon}{J} 
\right) \mp 2 \frac{h}{J} \sqrt{\frac{h^2}{J^2} - \frac{2 \varepsilon}{J}} .
\end{align*} 
If $h > 0$, then $m_-^2 > m_+^2 \implies |m_-| > |m_+|$, and, if $h < 0$, then 
$m_+^2 < m_-^2 \implies |m_+| > |m_-|$. The result follows by plugging in the 
values in the partition function.
\end{proof}

With the above computations, we can also now give a simpler definition of the set of allowed energies, i.e., of those values of $\vep$ for which the fixed energy 
ensemble is defined using \Cref{th:defcontfixedE}.
\begin{definition}
For $h\in \R$ and $\rho>0$, we define the set of possible energy densities 
$\mathcal{E}_{h,\rho}$ by
\begin{align*}
\mathcal{E}_{h,\rho} := \left\{ \varepsilon \in \mathbb{R} : \vep\le \frac{h^2}{2 J}, \ \left| \frac{|h|}{J} -
\sqrt{\frac{h^2}{J^2} - \frac{2 \varepsilon}{J}}\right| < \sqrt{\rho}\right\} .
\end{align*}
\end{definition}
We remark that for all of the above $h,\rho$ the set $\mathcal{E}_{h,\rho}$ contains $\vep=0$ and an interval of negative values of $\vep$.  In particular, $\mathcal{E}_{h,\rho}$ is non-empty.
Also, in case $h=0$, we have $\mathcal{E}_{0,\rho} = \left( - \frac{ \rho J}{2}, 0 \right]$.

\begin{theorem} \label{magdom2}
Let $h \not = 0$, $\rho>0$, and suppose $\vep\in\mathcal{E}_{h,\rho}$. For integrable functions $f : \mathcal{S} \to \mathbb{R}$ 
satisfying $\left< |f| \right>_{\text{MC}}^{m_{\pm}, \rho;N} \le K$ for some $K \ge 0$ and whenever the ensemble is defined, 
we have 
\begin{align}\label{eq:fdiffMC}
\left| \left< f \right>_\text{MC}^{\varepsilon, \rho;N}- \left< f 
\right>_{\text{MC}}^{m, \rho;N} \right| \leq 2 K \left| \frac{ \rho - \left( 
\frac{|h|}{J} + \sqrt{\frac{h^2}{J^2} - \frac{2 \varepsilon}{J}}\right)^2}{\rho 
- \left( \frac{|h|}{J} - \sqrt{\frac{h^2}{J^2} - \frac{2 
\varepsilon}{J}}\right)^2}\right|^{\frac{N - 3}{2}},
\end{align}
where
\begin{align*}
m = - \frac{h}{J} + \operatorname{sgn}(h) \sqrt{\frac{h^2}{J^2} - \frac{2 
\varepsilon}{J}} .
\end{align*}
In addition, 
\begin{align*}
\left< f \right>_\text{MC}^{\varepsilon, \rho;N} = \left< f 
\right>_{\text{MC}}^{m, \rho;N} + O(e^{- c N})
\end{align*}
for some positive constant $c > 0$.
\end{theorem}
\begin{proof}
If $\vep = \frac{h^2}{2 J}$, we have here $m=-\frac{h}{J}$, and
the Theorem is trivially true since then $\left< f \right>_\text{MC}^{\varepsilon, \rho;N}= \left< f \right>_{\text{MC}}^{m, \rho;N}$.  The same holds for those values 
where $\vep < \frac{h^2}{2 J}$ and $\rho\le \max(m_-^2,m_+^2)$.

In the remaining cases, we necessarily have $\vep < \frac{h^2}{2 J}$ and $m_-^2,m_+^2< \rho$.  Applying \cref{maxzlemma}, we have the following estimate
\begin{align*}
\left| \left< f \right>_\text{MC}^{\varepsilon, \rho;N}- \left< f 
\right>_{\text{MC}}^{m, \rho;N} \right| \leq (\left< |f| 
\right>_{\text{MC}}^{m_+, \rho;N} + \left< |f| \right>_{\text{MC}}^{m_-, 
\rho;N}) \frac{Z_{\text{MC}} (m_{- \operatorname{sgn}(h)}, 
\rho;N)}{Z_{\text{MC}} (m_{ \operatorname{sgn}(h)}, \rho;N)} .
\end{align*}
The results follow since the term inside the absolute values on the right hand side in (\ref{eq:fdiffMC}) is strictly less than one.
\end{proof}

If $h = 0$, there exists 
a suitable coupling which can be constructed from the couplings used for the 
fixed magnetization ensembles.
\begin{theorem} \label{energycouplingthm2}Let $h = 0$ and $\varepsilon', 
\varepsilon \in \left( - \frac{ \rho J}{2}, 0 \right]$. For $\varepsilon \not = 
0$, we have
\begin{align*}
w_2 (\mu_{\text{MC}}^{\varepsilon, \rho;N}, \mu_{\text{MC}}^{\varepsilon', 
\rho;N};N) \leq \frac{2}{J} 
\frac{1}{\sqrt{ - \frac{2 \varepsilon}{J}} 
\sqrt{1 - \left( - \frac{2\varepsilon}{J \rho}\right)}} |\varepsilon - \varepsilon'|,
\end{align*}
and, for $\varepsilon = 0$, we have
\begin{align*}
w_2 (\mu_{\text{MC}}^{0, \rho;N}, \mu_{\text{MC}}^{\varepsilon', \rho;N};N) \leq 
 \frac{2}{\sqrt{J}} |\varepsilon'|^{\frac{1}{2}}.
\end{align*}
\end{theorem}
\begin{proof}
Now $m_{\pm} = \pm\sqrt{\frac{-2\vep}{J}}$, and thus if $\varepsilon \not =  0$, 
we have $0<|m_{\pm}|<\rho$.  Therefore, these ensembles are defined as in the first case in \Cref{th:defcontfixedE}.

Suppose first that $\varepsilon,\varepsilon'<0$ and let 
$m_{\pm}$ and $m_{\pm}'$ be corresponding positive and negative 
magnetization densities to $\varepsilon$ and $\varepsilon'$, respectively. To obtain a transport map, we define $T : \mathbb{R} \times \mathbb{R}^{N - 
1} \to \mathbb{R} \times \mathbb{R}^{N - 1}$ by
\begin{align*}
T(z, \psi) := \mathbbm{1}(z \geq 0) \left(m_+' \sqrt{N}, \sqrt{(\rho - m_+'^2) 
N} \frac{\psi}{|| \psi ||_2} \right) + \mathbbm{1}(z < 0) \left(m_-' \sqrt{N}, 
\sqrt{(\rho - m_-'^2) N} \frac{\psi}{|| \psi ||_2} \right) .
\end{align*}
Now, let $U \in M_N (\mathbb{R})$ be the same unitary mapping as before. By 
setting $T' := U^{-1} \circ T \circ U$ and going through the same calculations 
as earlier, one can confirm that $\left< f \circ T' 
\right>_{\text{MC}}^{\varepsilon, \rho;N} = \left< f  
\right>_{\text{MC}}^{\varepsilon', \rho;N}$
for all observables $f$.  Thus $T'$ is a transport map and the associated coupling 
yields a bound
\begin{align*}
& w_2 (\mu_{\text{MC}}^{\varepsilon, \rho;N}, \mu_{\text{MC}}^{\varepsilon', 
\rho;N};N)^2
\le \frac{1}{2}\sum_{\sigma = \pm 1} \left[ (m_\sigma-m'_\sigma)^2
+ \left(\sqrt{\rho-m_\sigma^2}-\sqrt{\rho-(m'_\sigma)^2}\right)^2 \right]
\\ & \quad 
= \frac{2}{J} (\sqrt{-\vep}-\sqrt{-\vep'})^2
+ \left(\sqrt{\rho-\frac{-2 \vep}{J}}-\sqrt{\rho-\frac{-2 \vep'}{J}}\right)^2 
\\ & \quad \leq 
\frac{4}{J^2} \frac{\rho}{\left( - \frac{2 
\varepsilon}{J}\right) \left( \rho - \left( - \frac{2 
\varepsilon}{J}\right)\right)} (\varepsilon - \varepsilon')^2. 
\end{align*}

If $\varepsilon = 0>\vep'$, we have $m=0$ but it still holds that 
$\mean{f}_{\text{MC}}^{\varepsilon, \rho;N}
= \frac{1}{2}\mean{f}_{\text{MC}}^{m, \rho;N}
+\frac{1}{2}\mean{f}_{\text{MC}}^{m, \rho;N}$.  Proceeding as above then
yields an estimate
\begin{align*}
w_2 (\mu_{\text{MC}}^{0, \rho;N}, \mu_{\text{MC}}^{\varepsilon', \rho;N};N)^2 \leq 
 \frac{2}{J} (-\vep') + 
 \left(\sqrt{\rho}-\sqrt{\rho-\frac{-2 \vep'}{J}}\right)^2 
 \le
 \frac{2}{J} |\varepsilon'| + \frac{4}{\rho J^2} \varepsilon'^2 \le \frac{4}{J} |\varepsilon'| .
\end{align*}
The bound is also trivially true if $\vep=\vep'=0$.
Combining the above estimates proves the statement in the Theorem.
\end{proof}

\subsection{Canonical analysis}

Compared to the regular Curie-Weiss model from \cite{Kochma_ski_2013}, the canonical ensemble is somewhat simpler to analyse. In particular, one should note that when applying the identity
\begin{align*}
e^{\frac{x^2}{2}} = \frac{1}{\sqrt{2 \pi}} \int_{\mathbb{R}} dz \ e^{- \frac{z^2}{2} - xz} ,
\end{align*}
which is sometimes referred to as Gaussian linearization, to solve the partition function of the regular Curie-Weiss model, we are effectively adding an extra variable over which to integrate when using Laplace's method. To this end, in order to avoid multi-dimensional Laplace analysis, we can forego the use of Gaussian linearization, and use a direct $1$-dimensional Laplace method.

In the following definitions and calculations, there will be a significant difference in the treatment of the models depending on whether we are dealing with $h = 0$ or $h \not = 0$. In the previous section, for $h = 0$, we constructed an explicit coupling  between the fixed energy density ensembles. In the case of $h \not = 0$, it will turn out that the coupling between the fixed magnetization density ensembles is the more important object of study. This phenomenon seems to be closely related to the phase transitions in the mean-field spherical model which are thoroughly presented and analysed in \cite{Kastner2006}.

We will define the canonical ensembles with the help of the 
microcanonical ensembles.
\begin{definition}(Fluctuating magnetization and fixed particle 
density/auxiliary canonical ensemble) Let $\rho > 0$ and $\mu \in \mathbb{R}$. The 
auxiliary canonical ensemble with magnetic potential $\mu$ and particle density $\rho$ is 
defined via its action on bounded $1$-Lipschitz functions $f : \mathcal{S} \to 
\mathbb{R}$ by
\begin{align*}
 \left< f \right>_{\text{C}}^{\mu, \rho;N} &:= 
\frac{1}{\int_{-\sqrt{\rho}}^{\sqrt{\rho}} dm \ e^{- \mu N m} Z_\text{MC} (m, 
\rho;N)} \int_{-\sqrt{\rho}}^{\sqrt{\rho}} dm \ e^{- \mu N m} Z_\text{MC} (m, 
\rho;N) \left< f \right>_{\text{MC}}^{m, \rho;N} 
\\  &= 
\frac{1}{Z_\text{C}(\mu, \rho ;N)} \int_{-\sqrt{\rho}}^{\sqrt{\rho}} dm \ 
e^{- N(\mu m - s_\text{MC}(m, \rho))} (\rho - m^2)^{-\frac{3}{2}} \left< f 
\right>_{\text{MC}}^{m, \rho;N}  ,
\end{align*}
where we define the auxiliary canonical partition function by
\begin{align*}
Z_\text{C}(\mu, \rho ;N) := \int_{-\sqrt{\rho}}^{\sqrt{\rho}} dm \ e^{- N(\mu m 
- s_\text{MC}(m, \rho))} (\rho - m^2)^{-\frac{3}{2}}  \ ,
\end{align*}
and the specific auxiliary canonical free energy by
\begin{align*}
f_\text{C}(\mu, \rho;N) := -\frac{1}{N} \ln Z_\text{C}(\mu, \rho ;N) .
\end{align*}
\end{definition}
The case $h \not = 0$ is taken care of by the previous definition. We will refer 
to the special case of $h = 0$ as the fluctuating energy density ensemble. 
\begin{definition}[Fluctuating energy density and fixed particle density/canonical 
ensemble] Let $h = 0$. Suppose $\rho > 0$ and $\beta \in \mathbb{R}$. The canonical 
ensemble with inverse temperature $\beta$ and particle density $\rho$ is defined 
via its action on bounded $1$-Lipschitz functions $f : \mathcal{S} \to 
\mathbb{R}$ by
\begin{align*}
\left< f \right>_{\text{C}}^{\beta, \rho;N} &:= \frac{1}{\int_{- \frac{\rho 
J}{2}}^0 d \varepsilon \ e^{- \beta \varepsilon N} Z_\text{MC}(\varepsilon, \rho 
;N)}  \int_{- \frac{\rho J}{2}}^0 d \varepsilon \ e^{- \beta \varepsilon N} 
Z_\text{MC}(\varepsilon, \rho ;N) \left< f \right>_{\text{MC}}^{\varepsilon, 
\rho;N} \\ &= \frac{1}{Z_{\text{C}} (\beta, \rho;N)} 
\int_{- \frac{\rho J}{2}}^0 d 
\varepsilon \ e^{- N (\beta \varepsilon - s_{\text{MC}} (\varepsilon, \rho)) } 
\left( - \frac{2 \varepsilon}{J} \right)^{- \frac{1}{2}} \left( \rho + \frac{2 
\varepsilon}{J}\right)^{- \frac{3}{2}} \left< f 
\right>_{\text{MC}}^{\varepsilon, \rho;N} ,
\end{align*}
where the microcanonical partition function $Z_{\text{MC}}(\varepsilon, \rho;N)$ 
is defined by
\begin{align*}
Z_\text{MC}(\varepsilon, \rho;N) := \frac{\left(\rho + \frac{2 \varepsilon}{J} 
\right)^{\frac{N - 3}{2}}}{\sqrt{- \frac{2 \varepsilon}{J}}},
\end{align*}
 the specific microcanonical entropy $s_{\text{MC}} (\varepsilon, \rho)$ is 
defined by
\begin{align*}
s_{\text{MC}}(\varepsilon, \rho) := \frac{1}{2} \ln \left( \rho + \frac{2 
\varepsilon}{J} \right) ,
\end{align*}
and the canonical partition function $Z_{\text{C}} (\beta, \rho;N)$ 
by
\begin{align*}
Z_{\text{C}} (\beta, \rho;N) := \int_{- \frac{\rho J}{2}}^0 d \varepsilon \ e^{- 
N (\beta \varepsilon - s_{\text{MC}} (\varepsilon, \rho)) } \left( - \frac{2 
\varepsilon}{J} \right)^{- \frac{1}{2}} \left( \rho + \frac{2 
\varepsilon}{J}\right)^{- \frac{3}{2}},
\end{align*}
and the specific canonical free energy $f_{\text{C}} (\beta, \rho;N)$ by
\begin{align*}
f_{\text{C}} (\beta, \rho;N) := - \frac{1}{N} \ln Z_{\text{C}} (\beta, \rho;N) .
\end{align*}
\end{definition}

One can verify by formal calculations that this corresponds to the typical 
definition of an ensemble with a fixed average constraint.
We have overloaded the notation here 
similarly as was done in the previous section: for example, the functions 
$Z_{\text{C}} (\mu, \rho;N)$ and 
$Z_{\text{C}} (\beta, \rho;N)$ are different, but the name of the first parameter will uniquely determine to which we refer in the following.
Proceeding as before, we first present the asymptotics of the derivatives of the 
partition function.
\begin{theorem}\label{asympfunc3} Let $\rho > 0$ and $\mu \in \mathbb{R}$. Define $\psi_{\mu, \rho} : 
(-\sqrt{\rho},\sqrt{\rho}) \to \mathbb{R}$ by
\begin{align*}
\psi_{\mu, \rho} (m) := \mu m - \frac{1}{2} \ln (\rho - m^2) .
\end{align*}
Employing the shorthand notation
\[
 \mean{F(m)}' := \frac{1}{Z_\text{C}(\mu, \rho ;N)}\int_{-\sqrt{\rho}}^{\sqrt{\rho}} dm \ e^{- N \psi_{\mu, \rho} (m)} (\rho - m^2)^{- \frac{3}{2}} F(m),
\]
we have
\begin{align*}
\left< \frac{M}{N} \right>_{\text{C}}^{\mu, \rho;N} = 
\mean{m}'\quad\text{and}\quad \sigma_{\text{C}}^{\mu, \rho;N} \left( \frac{M}{N} \right) = \sqrt{\mean{m^2}'-(\mean{m}')^2}\,.
\end{align*}
Furthermore, if we fix $\rho > 0$, then for every $\mu \in \mathbb{R}$ there exists 
$m \in (-\sqrt{\rho},\sqrt{\rho})$ such that $\psi_{\mu, \rho}$ is minimized at $m$, and, for every $m \in (-\sqrt{\rho},\sqrt{\rho})$, there exists $\mu 
\in \mathbb{R}$ such that $\psi_{\mu, \rho}$ is minimized at $m$. The following 
asymptotics hold
\begin{align*}
\left< \frac{M}{N} \right>_{\text{C}}^{\mu, \rho;N} = \mathbbm{1}(\mu \not = 0) 
\left( \frac{1}{2 \mu} - \operatorname{sgn}(\mu) \sqrt{\left( \frac{1}{2 \mu} 
\right)^2 + \rho} \right) + O (N^{- \frac{1}{2}}), \quad \sigma_{\text{C}}^{\mu, 
\rho;N} \left( \frac{M}{N} \right) = O (N^{-\frac{1}{2}}) .
\end{align*}
\end{theorem}
\begin{proof}
The first part of the theorem follows directly by differentiating the specific 
free energies with respect to $\mu$ and dividing by the degrees of freedom $N$ 
appropriately.
 
Next, for fixed $\rho > 0$, we compute 
\begin{align*}
\psi'_{\mu, \rho} (m) = \mu + \frac{m}{\rho - m^2} , \quad \psi''_{\mu, \rho} (m) = 
\frac{1}{\rho - m^2} + \frac{2 m^2}{(\rho - m^2)^2} = \frac{\rho + m^2}{(\rho - 
m^2)^2} >0  .
\end{align*}
It follows that  the map $\psi_{\mu, \rho}$ is strictly concave for all $\mu \in \mathbb{R}$, and we can check that there is a 
unique global minimum at $m \in (-\sqrt{\rho},\sqrt{\rho})$ which satisfies 
$\psi'_{\mu, \rho} (m) = 0$. First, if $\mu = 0$, then clearly the minimizing $m = 
0$. If $\mu \not = 0$, we have
\begin{align*}
\psi_{\mu, \rho}' (m) = 0 \iff m = \frac{1}{2 \mu} \pm \sqrt{\left( \frac{1}{2 
\mu} \right)^2 + \rho} . 
\end{align*}
Next, if $\mu > 0$, then $
\frac{1}{2 \mu} + \sqrt{\left( \frac{1}{2 \mu} \right)^2 + \rho} > \sqrt{\rho} \ 
,
$
and thus the minimizing $m$ must be
\begin{align*}
m = \frac{1}{2 \mu} - \sqrt{\left( \frac{1}{2 \mu} \right)^2 + \rho}\ \in (-\sqrt{\rho},0) .
\end{align*}
If $\mu < 0$, then 
$
\frac{1}{2 \mu} - \sqrt{\left( \frac{1}{2 \mu} \right)^2 + \rho} = - 
\left(\sqrt{\left( \frac{1}{2 \mu} \right)^2 + \rho} - \frac{1}{2 \mu} \right) < 
-\sqrt{\rho} ,
$
and thus the minimizing $m$ must satisfy
\begin{align*}
m = \frac{1}{2 \mu} + \sqrt{\left( \frac{1}{2 \mu} \right)^2 + \rho}
\ \in (0,\sqrt{\rho}) .
\end{align*}
The conclusion is that, if $|m|<\sqrt{\rho}$, then 
\begin{align*}
\psi'_{\mu , \rho} (m) = 0 \iff m = \mathbbm{1}(\mu \not = 0) \left( \frac{1}{2 
\mu} - \operatorname{sgn}(\mu) \sqrt{\left( \frac{1}{2 \mu} \right)^2 + 
\rho}\right) .
\end{align*}
Furthermore, the above relation goes both ways. For every $\mu \in \mathbb{R}$ there 
exists a unique minimizing $m$ for the above equation, and, for every $m \in 
(-\sqrt{\rho},\sqrt{\rho})$, there exists $\mu \in \mathbb{R}$ such that the 
given $m$ is the minimizing term. This can be seen by simply studying the given 
equation above and considering the limits $|\mu| \to 0$ and $|\mu| \to \infty$ 
and using the continuity on the open intervals $(- \infty, 0)$ and $(0, 
\infty)$. 
 
The asymptotics of the average and standard deviation of magnetization density 
are given by the asymptotics of Laplace type integrals. We have
\begin{align*}
\left< \frac{M}{N} \right>_{\text{C}}^{\mu, \rho;N} = \mathbbm{1}(\mu \not = 0) 
\left( \frac{1}{2 \mu} - \operatorname{sgn}(\mu) \sqrt{\left( \frac{1}{2 \mu} 
\right)^2 + \rho} \right) + O (N^{- \frac{1}{2}}), \quad \sigma_{\text{C}}^{\mu, 
\rho;N} \left( \frac{M}{N} \right) = O (N^{-\frac{1}{2}}),
\end{align*}
as desired.
\end{proof}
Next, we present the asymptotics of the $h = 0$ case.
\begin{theorem} \label{asympfunc4} Let $h = 0$. Let $\rho > 0$ and $\beta \in \mathbb{R}$. Define 
$\psi_{\beta, \rho} : \left( -\frac{\rho J}{2}, 0 \right] \to \mathbb{R}$ by
\begin{align*}
\psi_{\beta, \rho} (\varepsilon) := \beta \varepsilon - \frac{1}{2} \ln \left( 
\rho + \frac{2 \varepsilon}{J}\right) .
\end{align*}
Employing the shorthand notation
\[
 \mean{F(\vep)}' := \frac{1}{Z_{\text{C}} (\beta, \rho;N)}\int_{- \frac{\rho 
J}{2}}^0 d \varepsilon \ e^{- N \psi_{\beta, \rho} (\varepsilon)} \left( - \frac{2 
\varepsilon}{J} \right)^{- \frac{1}{2}} \left( \rho + \frac{2 
\varepsilon}{J}\right)^{- \frac{3}{2}} F(\vep),
\]
we have
\begin{align*}
\left< \frac{H}{N} \right>_{\text{C}}^{\beta, \rho;N} = 
\mean{\vep}'\quad\text{and}\quad 
\sigma_{\text{C}}^{\beta, \rho;N} \left( \frac{H}{N} \right) = \sqrt{\mean{\vep^2}'-(\mean{\vep}')^2}\,.
\end{align*}
Furthermore, if we fix $\rho > 0$, then for every $\beta \geq \frac{1}{J \rho}$, there 
exists $\varepsilon \in \left( - \frac{\rho J}{2}, 0 \right]$ such that 
$\psi_{\beta, \rho}$ is minimized at $\varepsilon$, and, for every $\varepsilon 
\in \left( - \frac{\rho J}{2}, 0 \right]$, there exists $\beta \geq \frac{1}{J 
\rho}$ such that $\psi_\beta$ is minimized at $\varepsilon$. For such 
$\beta$ and $\varepsilon$, the following asymptotics holds
\begin{align}\label{eq:canfluct}
\left< \frac{H}{N} \right>_{\text{C}}^{\beta, \rho;N} = - \frac{J \rho}{2} 
\left( 1 - \frac{1}{\beta J \rho} \right) + O (N^{-\frac{1}{2}}), \quad 
\sigma_{\text{C}}^{\beta, \rho;N} \left( \frac{H}{N} \right) = O 
(N^{-\frac{1}{2}}) .
\end{align}

If $\beta < \frac{1}{J \rho}$, the mapping $\psi_\beta$ is always minimized 
at $0$, and the following asymptotics hold
\begin{align*}
\left< \frac{H}{N} \right>_{\text{C}}^{\beta, \rho;N} = O (N^{-1}), \quad 
\sigma_{\text{C}}^{\beta, \rho;N} \left( \frac{H}{N} \right) = O (N^{-1}) .
\end{align*}
\end{theorem}
\begin{proof}
The first part of this theorem follows directly by differentiating the specific free energies with respect to $\beta$ and dividing by the degrees of freedom N appropriately.

Next, fix $\rho>0$. We have
\begin{align*}
\psi_{\beta, \rho} ' (\varepsilon) = \beta - \frac{\frac{1}{J}}{\rho + \frac{2 
\varepsilon}{J}}, \quad \psi_{\beta, \rho} '' (\varepsilon) = \frac{2}{J^2} \frac{1}{ \left( 
\rho + \frac{2 \varepsilon}{J} \right)^2} .
\end{align*}
It follows that $\psi_{\beta, \rho}$ is strictly convex and obtains a unique global 
minimum when $\psi'_{\beta, \rho} (\varepsilon) = 0$. Computing it from the above, 
we see that $\psi'_{\beta, \rho} (\varepsilon) = 0 \iff \varepsilon = - \frac{J 
\rho}{2} \left(1 -  \frac{1}{\beta J \rho} \right)$. In particular, we see that 
for every $\varepsilon \in \left( - \frac{J \rho}{2}, 0 \right]$ there exists 
$\beta \ge \frac{1}{J \rho}$ such that the given $\varepsilon$ minimizes 
$\psi_{\beta, \rho}$, and, conversely, for every $\beta \ge \frac{1}{\rho J}$ there exists 
a minimizing value $\varepsilon \in \left( - \frac{\rho J}{2}, 0 \right]$. 
Furthermore, if $\beta < \frac{1}{\rho J}$, then $\psi'_{\beta, \rho}$ is strictly negative on 
the entire interval, and, as a result $\psi_{\beta, \rho}$ is minimized for $\varepsilon 
= 0$.
 
For the asymptotics, if $\beta \ge \frac{1}{J \rho}$, then the 
asymptotics are standard and we have
\begin{align*}
\left< \frac{H}{N} \right>_{\text{C}}^{\beta, \rho;N} = - \frac{J \rho}{2} 
\left( 1 - \frac{1}{\beta J \rho} \right) + O (N^{-\frac{1}{2}}), \quad 
\sigma_{\text{C}}^{\beta, \rho;N} \left( \frac{H}{N} \right) = O 
(N^{-\frac{1}{2}}) .
\end{align*}
If $\beta = \frac{1}{J \rho}$, then we need to choose half-integer values of ``$\alpha$'' in the Laplace method, but this will not alter the scaling of the asymptotics for the above ratios.
However, if $\beta < \frac{1}{J \rho}$, then $\psi'_{\beta, \rho} (\varepsilon) < 0$ for all $\varepsilon$, and since then ``$\mu=1$'' in the Laplace method, it follows that
\begin{align*}
\left< \frac{H}{N} \right>_{\text{C}}^{\beta, \rho;N} = O (N^{-1}), \quad 
\sigma_{\text{C}}^{\beta, \rho;N} \left( \frac{H}{N} \right) = O (N^{-1}) .
\end{align*}
This completes the proof of the Theorem.
\end{proof}

\subsection{Grand canonical analysis}

Finally, we will present the grand canonical ensemble and auxiliary grand canonical ensemble and the direct coupling method.  If one considers microcanonical to be the most fundamental ensemble, this will result in substantial simplification of computation of its expectation values in the thermodynamic limit since these can now be computed using the grand canonical ensemble which is a Gaussian measure.

\begin{definition}[Fluctuating magnetization and particle 
density/auxiliary grand canonical ensemble] Let $\mu \in \mathbb{R}$ and $\eta > 0$. The auxiliary grand canonical ensemble with magnetic potential $\mu$ and chemical potential $\eta$ is defined 
via its action on bounded $1$-Lipschitz functions $f : \mathcal{S} \to \mathbb{R}$ by
\begin{align*}
\left< f \right>_{\text{GC}}^{\mu, \eta;N} := 
\frac{1}{\int_{\mathbb{R}^{N}} d \phi \ e^{-\mu M[\phi] - \eta N[\phi]}}
\int_{\mathbb{R}^{N}} d \phi \ e^{-\mu M[\phi] -\eta N[\phi]} f(\phi)
 .
\end{align*}
\end{definition}
The definition may be rewritten using the same parametrization of the integrals as for the auxiliary microcanonical ensemble.  The result is summarized in the following Lemma.
\begin{lemma}
 Let $\mu \in \mathbb{R}$ and $\eta > 0$. 
 We have for all bounded $1$-Lipschitz functions $f : 
\mathcal{S} \to \mathbb{R}$
\begin{align*}
& \left< f \right>_{\text{GC}}^{\mu, \eta;N} = \frac{1}{Z_\text{GC}(\mu, \eta;N)} 
\int_{- \infty}^{\infty} dz \ e^{- \eta \left(z + \frac{\mu \sqrt{N}}{2 \eta} 
\right)^2} \int_0^\infty dr \ r^{N - 2} e^{- \eta r^2} 
\\ & \qquad \times
\frac{1}{|\mathbb{S}^{N - 2}|} \int_{\mathbb{S}^{N - 2}} d \Omega \ \left(f \circ U^{-1} \right) (z, 
r\Omega) ,
\end{align*}
where
\begin{align*}
Z_\text{GC}(\mu, \eta;N) := \int_{- \infty}^{\infty} dz \ e^{- \eta \left(z + 
\frac{\mu \sqrt{N}}{2 \eta} \right)^2} \int_0^\infty dr \ r^{N - 2} e^{- \eta 
r^2} .
\end{align*}
 \end{lemma}
We can now construct a direct coupling between the auxiliary microcanonical ensemble and 
the auxiliary grand canonical ensemble.
\begin{theorem} \label{directcoupling1}
Suppose $\rho > 0$, $m \in (- \sqrt{\rho}, 
\sqrt{\rho})$, $\mu \in \mathbb{R}$ and $\eta > 0$ satisfy the 
relations
\begin{align*}
m = - \frac{\mu}{2 \eta},\quad \rho = 
  \frac{1}{2 \eta} + \frac{\mu^2}{4 \eta^2} .
\end{align*}
Then,
\begin{align*}
w_2 (\mu_{\text{GC}}^{\mu, \eta;N}, \mu_{\text{MC}}^{m, \rho;N};N) \leq 
\frac{1}{\sqrt{\rho - m^2}} \frac{1}{\sqrt{N}},
\end{align*}
which implies
\begin{align*}
w_2 (\mu_{\text{GC}}^{\mu, \eta;N}, \mu_{\text{MC}}^{m, \rho;N};N) = O (N^{-\frac{1}{2}}) 
.
\end{align*}
\end{theorem}
\begin{proof}
Define $T : \mathbb{R} \times \mathbb{R}^{N - 1} \to \mathbb{R} \times 
\mathbb{R}^{N - 1}$ by
\begin{align*}
T(z, \psi) := \left(m \sqrt{N}, \sqrt{N(\rho - m^2)} \frac{\psi}{|| \psi ||_ 2} 
\right) ,
\end{align*}
and set $T' := U^{-1} \circ T \circ U$. It follows that $\left< f \circ T 
\right>_{\text{GC}}^{\mu, \eta;N} = \left< f \right>_{\text{MC}}^{m, \rho;N}$, and thus $T'$ is a transport map.  Therefore, using 
the related coupling we find an estimate
\begin{align*}
& w_2 (\mu_{\text{GC}}^{\mu, \eta;N}, \mu_{\text{MC}}^{m, \rho;N};N)^2 
 \\ & \quad  \leq 
\frac{1}{N} \frac{1}{\int_{- \infty}^{\infty} dz \ e^{- \eta \left(z + \frac{\mu 
\sqrt{N}}{2 \eta} \right)^2}} \int_{- \infty}^{\infty} dz \ e^{- \eta \left(z + 
\frac{\mu \sqrt{N}}{2 \eta} \right)^2} \left( z - m \sqrt{N} \right)^2 
 \\ & \qquad 
+ \frac{1}{N} \frac{1}{\int_{\mathbb{R}^{N - 1}} d \psi \ e^{- \eta || \psi ||^2} }\int_{\mathbb{R}^{N - 1}} d \psi e\ ^{- \eta || \psi ||^2} \left| || \psi || - \sqrt{N(\rho - m^2)} \right|^2 
 \\ & \quad = \frac{1}{N} \frac{1}{\int_{- \infty}^{\infty} dz \ e^{- \frac{ z^2}{2}}} 
\int_{- \infty}^{\infty} dz \ e^{- \frac{z^2}{2}} \left( \frac{z}{\sqrt{2 \eta}} 
- \left (m \sqrt{N} + \frac{\mu \sqrt{N}}{2 \eta} \right) \right)^2 
 \\ & \qquad + \frac{1}{N} \frac{1}{\int_{\mathbb{R}^{N - 1}} d \psi \ e^{-  \frac{|| \psi 
||^2}{2}} } \int_{\mathbb{R}^{N - 1}} d \psi \ e^{-  \frac{|| \psi ||^2}{2}} 
\left| \frac{|| \psi ||}{\sqrt{2 \eta}} - \sqrt{N(\rho - m^2)} \right|^2 .
\end{align*}
We compute
\begin{align*}
\frac{\mu \sqrt{N}}{2 \eta} = - m \sqrt{N}, \ \frac{1}{\sqrt{2 \eta}} = 
\sqrt{\rho - m^2} \iff \mu = - \frac{m}{\rho - m^2}, \ \eta = \frac{1}{2 (\rho - m^2)} .
\end{align*}
The converse result states that 
\begin{align*}
m = - \frac{\mu}{2 \eta}, \quad 
\rho = \frac{1}{2 \eta} + \frac{\mu^2}{4 \eta^2} .
\end{align*}

It follows that for every pair $(m, \rho)$ for which the auxiliary microcanonical ensemble 
exists, there exists a pair $(\mu, \eta)$ such that the auxiliary grand canonical ensemble 
exists, and, the converse result holds as well. For such a pair satisfying the 
equations given above, we have
\begin{align*}
\frac{1}{N} \frac{1}{\int_{- \infty}^{\infty} dz \ e^{- \frac{ z^2}{2}}} \int_{- 
\infty}^{\infty} dz \ e^{- \frac{z^2}{2}} \left( \frac{z}{\sqrt{2 \eta}} - \left 
(m \sqrt{N} + \frac{\mu \sqrt{N}}{2 \eta} \right) \right)^2 = \frac{1}{N} 
\frac{1}{2 \eta } = \frac{1}{N} \frac{1}{4 (\rho - m^2)} ,
\end{align*}
and
\begin{align*}
&\frac{1}{N} \frac{1}{\int_{\mathbb{R}^{N - 1}} d \psi \ e^{-  \frac{|| \psi 
||^2}{2}} } \int_{\mathbb{R}^{N - 1}} d \psi \ e^{-  \frac{|| \psi ||^2}{2}} 
\left| \frac{|| \psi ||}{\sqrt{2 \eta}} - \sqrt{N(\rho - m^2)} \right|^2  \\ &= 
\frac{1}{N} \frac{1}{2 \eta}\frac{1}{\int_{\mathbb{R}^{N - 1}} d \psi \ e^{-  
\frac{|| \psi ||^2}{2}} } \int_{\mathbb{R}^{N - 1}} d \psi \ e^{-  \frac{|| \psi 
||^2}{2}} \left| || \psi || - \sqrt{N} \right|^2 .
\end{align*}
We have $|| \psi ||^2 - N = 1+\sum_{i=1}^{N - 1} (\psi_i^2 - 1)$, and thus
\begin{align*}
\left( || \psi ||^2 - N \right)^2 = \sum_{i=1}^{N - 1} (\psi_i^2 - 1)^2 + 
\sum_{i \not = j}^{N - 1} (\psi_i^2 - 1) (\psi_j^2 - 1) - 2 \sum_{i=1}^{N - 1} 
(\psi_i^2 - 1) + 1 .
\end{align*}
Therefore,
\begin{align*}
& \left| || \psi || - \sqrt{N} \right|^2 = \frac{\left( || \psi ||^2 - N 
\right)^2}{ \left( || \psi || + \sqrt{N}\right)^2} 
\\ & \quad
\leq \frac{1}{N} \left( 
\sum_{i=1}^{N - 1} (\psi_i^2 - 1)^2 + \sum_{i \not = j}^{N - 1} (\psi_i^2 - 1) 
(\psi_j^2 - 1) - 2 \sum_{i=1}^{N - 1} (\psi_i^2 - 1) + 1\right) .
\end{align*}
It follows that
\begin{align*}
\frac{1}{\int_{\mathbb{R}^{N - 1}} d \psi \ e^{-  \frac{|| \psi ||^2}{2}} } 
\int_{\mathbb{R}^{N - 1}} d \psi \ e^{-  \frac{|| \psi ||^2}{2}} \left| || \psi 
|| - \sqrt{N} \right|^2 \leq \frac{2 N - 1}{N} . 
\end{align*}

Combining all the terms, we find
\begin{align*}
w_2 (\mu_{\text{GC}}^{\mu, \eta;N}, \mu_{\text{MC}}^{m, \rho;N};N)^2 \leq 
\frac{1}{N} \frac{1}{2 \eta} \frac{3N - 1}{N} = \frac{1}{4 (\rho - m^2)} 
\frac{3N - 1}{N^2} \le \frac{1}{(\rho - m^2) N},
\end{align*}
which implies the bound stated in the Theorem.  
\end{proof}

If $h\ne 0$, the microcanonical energy ensemble is well-approximated by an auxiliary microcanonical magnetization ensemble whose auxiliary grand canonical theory we already covered above.  
For the case of $h = 0$, we consider the following grand canonical energy ensembles.
\begin{definition}[Fluctuating energy and particle density/grand canonical ensemble] Suppose $h=0$, $\mu > 0$ and $\beta < \frac{2 \mu}{J}$.
The grandcanonical ensemble with inverse temperature $\beta$ and chemical potential $\mu$ is defined 
via its action on bounded $1$-Lipschitz functions $f : \mathcal{S} \to \mathbb{R}$ by
\begin{align*}
\left< f \right>_{\text{GC}}^{\beta, \mu;N} := 
\frac{1}{\int_{\mathbb{R}^{N}} d \phi \ e^{-\beta H[\phi] - \mu N[\phi]}}
\int_{\mathbb{R}^{N}} d \phi \ e^{-\beta H[\phi] - \mu N[\phi]} f(\phi)
 .
\end{align*}
\end{definition}
The definition may be rewritten using the same parametrization of the integrals as for the microcanonical ensemble.  The result is summarized in the following Lemma.
\begin{lemma}
Let $h = 0$. Let $\mu > 0$ and $\beta < \frac{2 \mu}{J}$. Then
\begin{align*}
\left< f \right>_{\text{GC}}^{\beta, \mu;N} = \frac{1}{Z_\text{GC} (\beta, 
\mu;N)}\int_{- \infty}^\infty dz \ e^{- \left( \mu - \frac{\beta J}{2} 
\right)z^2} \int_{0}^\infty dr \ r^{N - 2} e^{- \mu r^2}   
\frac{1}{|\mathbb{S}^{N - 2}|}\int_{\mathbb{S}^{N - 2}} d \Omega \ \left(f \circ 
U^{-1} \right) (z, r\Omega),
\end{align*}
where
\begin{align*}
Z_\text{GC} (\beta, \mu;N): = \int_{- \infty}^\infty dz \ e^{- \left( \mu - 
\frac{\beta J}{2} \right)z^2} \int_{0}^\infty dr \ r^{N - 2} e^{- \mu r^2} .
\end{align*}
 \end{lemma}

For the fixed energy density ensemble, there is only a single value of energy 
density for which a direct coupling can be constructed.
\begin{theorem}
Suppose $h=0$ and $\mu > 0$ are given.  Define $\rho = 
\frac{1}{2 \mu}$. Then, for all $\beta < \frac{2 \mu}{J}$, we have
\begin{align*}
w_2 (\mu_{\text{GC}}^{\beta, \mu;N}, \mu_{\text{MC}}^{\vep, \rho;N}|_{\vep=0};N)^2 \leq 
\frac{1}{N} \frac{1}{2 \left(\mu - \frac{\beta J}{2} \right)} + \frac{1}{N} \frac{1}{\mu} = \frac{1}{N}\left(  \frac{\rho}{1 - \rho \beta J} + 2 \rho \right),
\end{align*}
implying
\begin{align*}
w_2 (\mu_{\text{GC}}^{\beta, \mu;N}, \mu_{\text{MC}}^{\vep , \rho;N}|_{\vep=0};N) = O (N^{-\frac{1}{2}}) 
.
\end{align*}
\end{theorem}
\begin{proof}
Let us begin by considering the more general case with $h\in \R$ and $\mu,\rho>0$ arbitrary.
Let $m_+$ and $m_-$ be the corresponding negative and positive magnetization 
densities to the given $\varepsilon$. Define $T : \mathbb{R} \times 
\mathbb{R}^{N - 1} \to \mathbb{R} \times \mathbb{R}^{N - 1}$ by
\begin{align*}
T(z, \psi) := \mathbbm{1}(z \geq 0) \left(m_+ \sqrt{N}, \sqrt{N(\rho - m_+^2)} 
\frac{\psi}{|| \psi ||_ 2} \right) + \mathbbm{1}(z < 0) \left(m_- \sqrt{N}, 
\sqrt{N(\rho - m_-^2)} \frac{\psi}{|| \psi ||_ 2} \right) .
\end{align*}
Define $T' := U^{-1} \circ T \circ U$. It follows that $\left< f \circ T' 
\right>_{\text{GC}}^{\beta, \mu;N} = \left< f \right>_{\text{MC}}^{\varepsilon,\rho;N}$, and $T'$ is a transport map.  Using the associated coupling, we find
\begin{align*}
& w_2 (\mu_{\text{GC}}^{\beta, \mu;N}, \mu_{\text{MC}}^{\varepsilon\rho;N};N)^2 
\leq \frac{1}{N} \frac{1}{\int_{0}^\infty dz \ e^{- \left( \mu - \frac{\beta J}{2} \right)z^2}} \int_{0}^\infty dz \ e^{- \left( \mu - \frac{\beta J}{2} 
\right)z^2} \left( z - \sqrt{- \frac{2 \varepsilon}{J}}  \sqrt{N} \right)^2 
\\ & \qquad + \frac{1}{N} \frac{1}{\int_{\mathbb{R}^{N - 1}} d \psi \ e^{- \mu || \psi 
||^2} }\int_{\mathbb{R}^{N - 1}} d \psi e\ ^{- \mu || \psi ||^2} \left| || \psi 
|| - \sqrt{N \left(\rho - \left( - \frac{2 \varepsilon}{J}\right) \right)} 
\right|^2 
\\ & \quad = \frac{1}{N} \frac{1}{\int_{0}^{\infty} dz \ e^{- \frac{ z^2}{2}}} 
\int_{0}^{\infty} dz \ e^{- \frac{z^2}{2}} \left( \frac{z}{\sqrt{2 \left( \mu - 
\frac{\beta J}{2}\right)}} - \sqrt{- \frac{2 \varepsilon}{J}}  \sqrt{N} 
\right)^2 
\\ & \qquad + \frac{1}{N} \frac{1}{\int_{\mathbb{R}^{N - 1}} d \psi \ e^{-  \frac{|| \psi 
||^2}{2}} } \int_{\mathbb{R}^{N - 1}} d \psi \ e^{-  \frac{|| \psi ||^2}{2}} 
\left| \frac{|| \psi ||}{\sqrt{2 \mu}} - \sqrt{N \left(\rho - \left( - \frac{2 
\varepsilon}{J}\right) \right)} \right|^2 .
\end{align*}

Note that for $\varepsilon \not = 0$, we have
\begin{align*}
\frac{1}{N} \frac{1}{\int_{0}^{\infty} dz \ e^{- \frac{ z^2}{2}}} 
\int_{0}^{\infty} dz \ e^{- \frac{z^2}{2}} \left( \frac{z}{\sqrt{2 \left( \mu - 
\frac{\beta J}{2}\right)}} - \sqrt{- \frac{2 \varepsilon}{J}}  \sqrt{N} 
\right)^2  \sim -\frac{2 \vep}{J} ,
\end{align*}
which does not provide any additional convergence for the local expectation error estimates.  However, under the assumptions listed in the Theorem, i.e., if $h=0=\vep$, $\mu > 0$, $\rho = \frac{1}{2 \mu}$,
we find via the same computation as above that 
\begin{align*}
w_2 (\mu_{\text{GC}}^{\beta, \mu;N}, \mu_{\text{MC}}^{\varepsilon, \rho;N};N)^2 
\leq \frac{1}{N} \frac{1}{2 \left(\mu - \frac{\beta J}{2} \right)} + \frac{1}{N} 
\frac{1}{2 \mu} 2 = \frac{1}{N}\left(  \frac{\rho}{\rho - \beta 
J} + \rho  \right) .
\end{align*}
Note that the above holds for all $\beta < \frac{2 \mu}{J} \iff \beta <  
\frac{1}{\rho J}$.
\end{proof}

For the cases $\beta \geq \frac{1}{\rho J}$, we must introduce another class of 
auxiliary measures.
\begin{definition} Let $\overline{\mu} \geq 0$ and $\eta > 0$. We define an
alternate auxiliary grand canonical ensemble with parameters $\overline{\mu}$ and $\eta$ 
via its action on bounded $1$-Lipschitz functions $f : \mathcal{S} \to 
\mathbb{R}$ by
\begin{align*}
\left< f \right>_{\text{AGC}}^{\overline{\mu}, \eta;N} :=  
\frac{1}{Z_{\text{AGC}} (\overline{\mu}, \eta;N)} \int_{\mathcal{S}} d \phi \, 
e^{- \eta || \phi ||^2} \cosh (\overline{\mu}M[ \phi] ) f(\phi),
\end{align*} 
where
\begin{align*}
Z_{\text{AGC}} (\overline{\mu}, \eta;N) = \int_{\mathcal{S}} d \phi \, e^{- \eta || \phi ||^2} \cosh (\overline{\mu}M[ \phi] )  .
\end{align*}
\end{definition}

By direct computation, we also have then
\begin{align*}
\left< f \right>_{\text{AGC}}^{\overline{\mu}, \eta;N} = \frac{1}{2} \left< f 
\right>_{\text{GC}}^{\overline{\mu}, \eta;N} + \frac{1}{2} \left< f 
\right>_{\text{GC}}^{-\overline{\mu}, \eta;N} .
\end{align*}
Comparing this to the definition of the microcanonical ensemble, 
we find  using Lemma \ref{directcoupling1} that if $h= 0$ and $-\frac{\rho J}{2}<\vep<0$, then, with $\overline{\mu} = \sqrt{\frac{-2\vep}{J}}\frac{1}{\rho+\frac{2\vep}{J}}$ and $\eta = \frac{J}{2 (J\rho+2\vep)}$, 
\begin{align*}
w_2 (\mu_{\text{AGC}}^{\overline{\mu}, \eta;N}, \mu_{\text{MC}}^{\vep, \rho;N};N) = O (N^{-\frac{1}{2}}) . 
\end{align*}

One should note that there is no direct coupling of this new alternate auxiliary grand 
canonical ensemble to the microcanonical ensemble because the probability 
measures are not disjoint. However, the individual grand canonical ensembles do 
converge suitably to the fixed magnetization density case, and thus we still 
have the desired local convergence properties. We also remark that the case 
$\overline{\mu} = 0$ corresponds to the regular grand canonical ensemble given by a Gaussian measure with $\beta=0$.

\subsection{Convergence of finite marginal distributions and finite moments}

In this subsection, we will collect and apply the upper bounds and error estimates presented for the continuum model to formulate the main local convergence theorems. Since our main goal is to prove convergence theorems and error estimates for local observables of the mean-field spherical microcanonical ensemble, we will present a variety of target measures which the local observables can converge to. Some of these target measures will come from thermodynamic ensembles and others from auxiliary ensembles. In particular, the value of the constant $h$ will have a significant impact on the choice of target measure. 

The main theorems concerning the convergence of moments required the boundedness 
of single moments of all degrees. To this end, we will employ the following lemma.
\begin{lemma} \label{mombound1} Let $f : \mathbb{R} \to \mathbb{R}$ be integrable with respect to any Gaussian measure. Let $x \in \Lambda$ and 
define $P_x : \mathcal{S} \to \mathbb{R}$ by $P_x (\phi) = \phi_x$. It 
follows that for all $\rho > 0$ and $m \in (- \sqrt{\rho}, \sqrt{\rho})$, we 
have
\begin{align*}
\left< f \circ P_x \right>_{\text{MC}}^{m, \rho;N} = O(1)  .
\end{align*}
\end{lemma}
\begin{proof}
A direct calculation using the delta function definition of the measures shows 
that
\begin{align*}
&\left< f \circ P_x \right>_{\text{MC}}^{m, \rho;N} \\ &= \frac{1}{C(m, 
\rho;N)} \int_{- \infty}^\infty d \phi_x \ f (\phi_x) \mathbbm{1} \left( (\phi_x 
- m)^2 \leq (\rho - m^2) (N - 1) \right)  \left( 1  - \frac{(\phi_x - 
m)^2}{(\rho - m^2) (N - 1)}  \right)^{\frac{N - 4}{2}} ,
\end{align*}
where
\begin{align*}
C(m, \rho;N) := \int_{- \infty}^\infty d \phi_x \ \mathbbm{1} \left( (\phi_x - 
m)^2 \leq (\rho - m^2) (N - 1) \right)  \left( 1  - \frac{(\phi_x - m)^2}{(\rho 
- m^2) (N - 1)}  \right)^{\frac{N - 4}{2}}.
\end{align*}
For $N \geq 5$, we have
\begin{align*}
&\mathbbm{1} \left( (\phi_x - m)^2 \leq (\rho - m^2) (N - 1) \right)  \left( 1  
- \frac{(\phi_x - m)^2}{(\rho - m^2) (N - 1)}  \right)^{\frac{N - 4}{2}} \\ 
&\leq \mathbbm{1} \left( (\phi_x - m)^2 \leq (\rho - m^2) (N - 1) \right)  e^{- 
\frac{(\phi_x - m)^2}{2 (\rho - m^2)} \frac{N - 4}{N - 1}}  \\ &\leq e^{- 
\frac{(\phi_x - m)^2}{8 (\rho - m^2)}} .
\end{align*}
By the dominated convergence theorem, using the assumed integrability of $f$, we have
\begin{align*}
\lim_{N \to \infty} \left< f \circ P_x \right>_{\text{MC}}^{m, \rho;N} = 
\frac{1}{\int_{- \infty}^\infty d \phi_x \ e^{- \frac{(\phi_x - m)^2}{2 (\rho - 
m^2)}}} \int_{- \infty}^\infty d \phi_x \ f(\phi_x) e^{- \frac{(\phi_x - m)^2}{2 
(\rho - m^2)}} < \infty \ ,
\end{align*}
which implies that 
\begin{align*}
\left< f \circ P_x \right>_{\text{MC}}^{m, \rho;N} = O(1) .
\end{align*}
\end{proof}

\begin{lemma} \label{mombound2}
Let $f : \mathbb{R} \to \mathbb{R}$ be integrable with respect to any 
Gaussian measure. 
Let $x \in \Lambda$ and define $P_x : \mathcal{S} \to 
\mathbb{R}$ by $P_x (\phi) = \phi_x$. Then for all $\rho > 0$, $h\in \R$, and 
$\varepsilon \in \mathcal{E}_{h,\rho}$, we have
\begin{align*}
\left< f \circ P_x \right>_{\text{MC}}^{\varepsilon, \rho;N} = O(1)  .
\end{align*}
\end{lemma}
\begin{proof}
By definition, 
$\left< f \circ P_x \right>_{\text{MC}}^{\varepsilon, \rho;N}$
is either equal to one of the expectations studied in the previous Lemma, or it is a convex combination of 
$ \left< f \circ P_x \right>_{\text{MC}}^{m_+, \rho;N}$ and 
$ \left< f \circ P_x \right>_{\text{MC}}^{m_-, \rho;N}$.
In  both of these cases, the result remains bounded as $N\to \infty$.
\end{proof}
As we remarked earlier, the phase transitions in the mean-field spherical model result in the need for different limiting measures outside of the standard ensembles when using the coupling method. Of particular importance is the parameter $h \in \mathbb{R}$. In the following convergence results, we will always explicitly state for which different parameters and limiting measures the convergence results hold.

First, we will state the convergence result for the auxiliary microcanonical, auxiliary canonical, microcanonical, and canonical ensembles.
\begin{theorem} \label{mccequiv1}
Let $\rho > 0$ and $m \in (- \sqrt{\rho}, \sqrt{\rho})$ and define $\mu := - 
\frac{m}{\rho - m^2}$. Let $I \subset \Lambda$ be  a fixed size index set, and let $f 
: \mathbb{R}^{|I|} \to \mathbb{R}$ be a bounded $1$-Lipschitz function with 
respect to the $|| \cdot ||_2$-norm. It follows that
\begin{align*}
\left< f \circ P_I \right>_{\text{MC}}^{m, \rho;N} = \left< f \circ P_{I} 
\right>_{\text{C}}^{\mu, \rho; N} + O (N^{- \frac{1}{2}}) .
\end{align*}
\end{theorem}
\begin{proof}
The result follows by applying the free energy coupling presented in 
\Cref{freeenergylip}, along with the $w_2$ bound presented in 
\Cref{magcouplingthm2}, and with the asymptotics presented in \Cref{asympfunc3}. 
\end{proof}
\begin{theorem} \label{dommcequiv1}
Suppose $h \not = 0$ and $\rho > 0$.  Assume $\varepsilon \in \mathcal{E}_{h,\rho}$. Define $m := 
- \frac{h}{J} + \operatorname{sgn}(h) \sqrt{\frac{h^2}{J^2} - \frac{2 
\varepsilon}{J}}$ and $\mu := - \frac{m}{\rho - m^2}$. Let $I \subset \Lambda$ 
be  a fixed size index set, and let $f : \mathbb{R}^{|I|} \to \mathbb{R}$ be a 
bounded $1$-Lipschitz function with respect to the $|| \cdot ||_2$-norm. It 
follows that
\begin{align*}
\left< f \circ P_I \right>_{\text{MC}}^{\varepsilon, \rho;N} = \left< f \circ 
P_{I} \right>_{\text{C}}^{\mu, \rho; N} + O (N^{- \frac{1}{2}}) .
\end{align*}
\end{theorem}
\begin{proof}
By \Cref{magdom2}, we have
\begin{align*}
\left< f \circ P_I \right>_{\text{MC}}^{\varepsilon, \rho;N} = \left< f \circ 
P_I \right>_{\text{MC}}^{m, \rho;N} + O (e^{- c N}),
\end{align*}
and, by \Cref{mccequiv1}, we have
\begin{align*}
\left< f \circ P_I \right>_{\text{MC}}^{m, \rho;N} = \left< f \circ P_{I} 
\right>_{\text{C}}^{\mu, \rho; N} + O (N^{- \frac{1}{2}}) ,
\end{align*}
from which the result follows.
\end{proof}

\begin{theorem}
Let $h = 0$.  Suppose $\rho > 0$ and $\varepsilon \in \left( - \frac{J \rho}{2}, 0 
\right]$. Let $I \subset \Lambda$ be  a fixed size index set, and let $f : 
\mathbb{R}^{|I|} \to \mathbb{R}$ be a bounded $1$-Lipschitz function with 
respect to the $|| \cdot ||_2$-norm. For $\varepsilon \in \left( - \frac{J 
\rho}{2}, 0 \right)$, define $\beta := \frac{\frac{1}{J}}{\rho + \frac{2 
\varepsilon}{J}}$. It follows that
\begin{align*}
\left< f \circ P_I \right>_{\text{MC}}^{\varepsilon, \rho;N} = \left< f \circ 
P_{I} \right>_{\text{C}}^{\beta, \rho; N} + O (N^{- \frac{1}{2}}) .
\end{align*}
For $\varepsilon = 0$, let $\beta < \frac{1}{\rho J}$ be arbitrary. It follows that 
\begin{align*}
\left< f \circ P_I \right>_{\text{MC}}^{0, \rho;N} = \left< f \circ P_{I} 
\right>_{\text{C}}^{\beta, \rho; N} + O (N^{- \frac{1}{2}}) .
\end{align*}
\end{theorem}
\begin{proof}
If $\varepsilon \in \left( - \frac{J \rho}{2}, 0 \right)$, the result follows by 
applying the free energy coupling presented in \Cref{freeenergylip}, along with 
the $w_2$ bound presented in \Cref{energycouplingthm2}, and the asymptotics 
presented in \Cref{asympfunc4}.
 
If $\varepsilon = 0$, then observe that the $w_2$ bound in 
\Cref{energycouplingthm2} is not Lipschitz in the appropriate sense to directly 
apply \Cref{freeenergylip}. However, following the proof of 
\Cref{freeenergylip}, we can apply the following inequality
\begin{align*}
\left< w_2 \left( \mu_{\text{MC}}^{0, \rho;N}, \mu_{\text{MC}}^{\frac{H}{N}, 
\rho;N};N\right)\right>_{\text{C}}^{\beta, \rho;N} \leq
 \frac{2}{\sqrt{J}} 
 \left(\left< -H/N \right>_{\text{C}}^{\beta, \rho; N}\right)^{\frac{1}{2}} ,
\end{align*}
where the upper-index $\frac{H}{N}$ is a non-positive random variable of the canonical 
ensemble. It follows that
\begin{align*}
\left| \left< f \circ P_I \right>_{\text{MC}}^{0, \rho;N} - \left< f \circ 
P_{I} \right>_{\text{C}}^{\beta, \rho; N} \right| \leq C 
 \left(\left< -H/N \right>_{\text{C}}^{\beta, \rho; N}\right)^{\frac{1}{2}},
\end{align*}  
for a global constant $C > 0$.  By considering the asymptotics presented in 
\Cref{asympfunc4}, we have
\begin{align*}
\left< f \circ P_I \right>_{\text{MC}}^{0, \rho;N} = \left< f \circ P_{I} 
\right>_{\text{C}}^{\beta, \rho; N} + O(N^{- \frac{1}{2}}) .
\end{align*}
\end{proof}
Finally, we will state the convergence result for the auxiliary microcanonical, auxiliary grand canonical, alternate auxiliary grand canonical, microcanonical, and grand canonical ensembles.
\begin{theorem} \label{mcgcequiv1}
Let $\rho > 0$ and $m \in (- \sqrt{\rho}, \sqrt{\rho})$ and define $\mu := - 
\frac{m}{\rho - m^2}$ and $\eta := \frac{1}{2 (\rho - m^2)}$. Let $I \subset 
\Lambda$ be  a fixed size index set, and let $f : \mathbb{R}^{|I|} \to \mathbb{R}$ 
be a bounded $1$-Lipschitz function with respect to the $|| \cdot ||_2$-norm or 
let $f$ be a finite product of finite order moments. It follows that
\begin{align*}
\left< f \circ P_I \right>_{\text{MC}}^{m, \rho;N} = \left< f \circ P_{I} 
\right>_{\text{GC}}^{\mu, \eta; N} + O (N^{- \frac{1}{2}}) .
\end{align*}
\end{theorem}
\begin{proof}
The result follows by applying the direct coupling method in \Cref{wass2est} 
along with the $w_2$ bound given in \Cref{directcoupling1}. The convergence of 
the finite dimensional moments follows from \Cref{wass2momest} and the fact that 
the moments of both relevant ensembles are bounded by \Cref{mombound1}.
\end{proof}
\begin{theorem}
Suppose $h \not = 0$ and $\rho > 0$.  Assume $\varepsilon \in \mathcal{E}_{h,\rho}$. Define $m := 
- \frac{h}{J} + \operatorname{sgn}(h) \sqrt{\frac{h^2}{J^2} - \frac{2 
\varepsilon}{J}}$ and set then  $\mu := - \frac{m}{\rho - m^2}$ and $\eta := 
\frac{1}{2 (\rho - m^2)}$. Let $I \subset \Lambda$ be  a fixed size index set, and 
let $f : \mathbb{R}^{|I|} \to \mathbb{R}$ be a bounded $1$-Lipschitz function 
with respect to the $|| \cdot ||_2$-norm or let $f$ be a finite product of finite 
order moments. It follows that
\begin{align*}
\left< f \circ P_I \right>_{\text{MC}}^{\varepsilon, \rho;N} = \left< f \circ 
P_{I} \right>_{\text{GC}}^{\mu, \eta; N} + O (N^{- \frac{1}{2}}) .
\end{align*}
\end{theorem}
\begin{proof}
The proof follows from  \Cref{mcgcequiv1} via the same steps as in 
the proof of \Cref{dommcequiv1}. 
\end{proof}
\begin{theorem}
Suppose $h = 0$ and $\rho > 0$.  Assume $\varepsilon \in \left( - \frac{J \rho}{2}, 0 
\right]$ and define $\overline{\mu} := \frac{\sqrt{- \frac{2 
\varepsilon}{J}}}{\rho - \left( - \frac{2 \varepsilon}{J} \right)}$ and $\eta := 
\frac{1}{2 \left( \rho - \left( - \frac{2 \varepsilon}{J} \right) \right)}$. Let 
$I \subset \Lambda$ be  a fixed size index set, and let $f : \mathbb{R}^{|I|} \to 
\mathbb{R}$ be a bounded $1$-Lipschitz function with respect to the $|| \cdot 
||_2$-norm or a finite product of finite order moments. It follows that
\begin{align*}
\left< f \circ P_I \right>_{\text{MC}}^{\varepsilon, \rho;N} = \left< f \circ 
P_I \right>_{\text{AGC}}^{\overline{\mu}, \eta;N} + O (N^{-\frac{1}{2}}) .
\end{align*}
\end{theorem}
\begin{proof}
The result follows by splitting the fixed energy density ensemble into its fixed 
magnetization ensity ensembles and applying \Cref{mcgcequiv1}.
\end{proof}

\subsection{Remark on choice of cost function}

For this model, it should be observed that the $w_2$ convergence is a natural 
choice of convergence from the perspective that it implies both $w_1$ and $w_2$ 
convergence simultaneously, which in turn implies that the magnetization density 
converges along with the energy density. Without this property, a coupling of 
suitable strength between the microcanonical and grand canonical measures seems 
unlikely. Observe that
\begin{align*}
w_1 \left( \mu_{\text{MC}}^{m, \rho;N}, \mu_{\text{MC}}^{m', \rho;N};N\right) 
\leq  w_2 \left( \mu_{\text{MC}}^{m, \rho;N}, \mu_{\text{MC}}^{m', \rho;N} ;N\right)  .
\end{align*}
Employing the lower bound in the triangle inequality, we have
\begin{align*}
w_1 \left( \mu_{\text{MC}}^{m, \rho;N}, \mu_{\text{MC}}^{m', \rho;N}\right) \geq 
|m - m'| .
\end{align*}
Now, if we consider \Cref{magcouplingthm2}, then we have
\begin{align*}
|m - m'| \leq w_2 \left( \mu_{\text{MC}}^{m, \rho;N}, \mu_{\text{MC}}^{m', 
\rho;N}; N\right) \leq 
\left( 1 
+ \frac{2}{\sqrt{1-m^2/\rho}} \right) |m - m'| .
\end{align*}
Thus, even though optimality of the transport was not necessarily achieved as in 
the discrete case, the best possible scaling in the dependence on changes in the parameter $m$ was still obtained here.

\section*{Acknowledgements}

We are grateful to Stefan Gro\ss kinsky and Eero Saksman 
for their comments on coupling techniques used in stochastic particle systems.  
We also thank Herbert Spohn for several discussions and references about 
equivalence of ensembles in general.  We also thank the anonymous reviewer for helpful comments and references.

The work has been supported by the 
Academy of Finland via the Centre of Excellence in Analysis and Dynamics 
Research (project 307333), and it has 
also benefited from the support of the project EDNHS ANR-14-CE25-0011 of the 
French National Research Agency (ANR).

\appendix

\section{Rigorous asymptotic analysis of Laplace-type integrals}\label{sec:Laplace}

First, we will fix some notation and definitions.
\begin{definition} [Asymptotic equivalence] Let $f,g : \mathbb{R} \to 
\mathbb{R}$ be suitable functions so that the following limits and quotients 
exist. We say that $f$ and $g$ are asymptotically equivalent at $a \in 
\overline{\mathbb{R}}:=[-\infty,\infty]$ if 
\begin{align*}
\lim_{x \to a} \frac{f(x)}{g(x)} = 1 .
\end{align*}
Asymptotic equivalence will be denoted $f \sim g$ without reference to the 
limiting point $a$ if it is clear from context.

Furthermore, we say that a function $f$ admits an asymptotic power series 
representation at a point $a \in \mathbb{R}$ if there is a sequence of constants 
$(a_k)_{k\in \N_0}$ and some $\mu \geq 0$ such that
\begin{align*}
\lim_{x\to a} (x-a)^{-N-\mu}\left(f(x) -\sum_{k=0}^{N} a_k (x - a)^{k + 
\mu}\right)
= 0 
\end{align*}
for all $N\in \N_0$.  Since this implies that $f(x) -\sum_{k=0}^{N-1} a_k (x - 
a)^{k + \mu}\sim a_N  (x - a)^{N + \mu}$ whenever $a_N\ne 0$,
we will use the notation
\begin{align*}
f(x)  \sim \sum_{k=0}^\infty a_k (x - a)^{k + \mu}
\end{align*}
to denote the above without reference to $N$, even if the power series on the 
right does not converge.

Analogously, we define asymptotic power series representation as $x\to \infty$ 
by requiring that $x\mapsto f(1/x)$, $x>0$, has an asymptotic power series at 
$0$.
Explicitly, we then require
\begin{align*}
\lim_{x\to \infty} x^{N+\mu}\left(f(x) -\sum_{k=0}^{N} a_k x^{-k - \mu}\right)
= 0 \,,
\end{align*}
for all $N\in \N_0$, and denote this by 
\begin{align*}
f(x)  \sim \sum_{k=0}^\infty a_k x^{-k -\mu}
\end{align*}
\end{definition}

The asymptotic analysis of Laplace-type integrals has been studied  extensively. 
For completeness, we will present below a general form of the asymptotics of 
Laplace-type integrals.
\begin{theorem} \label{asymptotic1}
Let $h : [a,b] \to \mathbb{R}$ be a function satisfying the following 
conditions: 
\begin{itemize}
\item $h$ attains a global minimum at the left end-point $a$, $h(x) > h(a)$ for 
all $x \in (a,b)$, and for every $\delta > 0$ we have $\inf_{x \in [a+ \delta, 
b)} \{ h(x) - h(a) \} > 0$.
\item $h$ admits a power series representation at the left end-point $a$ of the 
form
\begin{align*}
h(x) \sim h(a) + \sum_{s=0}^\infty a_s (x - a)^{s + \mu}
\end{align*}
for some $\mu > 0$ and with $a_0 \not = 0$.
\item $h$ is differentiable in a neighbourhood of $a$ and the previous power 
series representation can be term-wise differentiated to give
\begin{align*}
h'(x) \sim \sum_{s=0}^\infty a_s (s +  \mu) (x - a)^{s + \mu - 1}\,.
\end{align*}
\item $h'$ is continuous in a neighbourhood of $a$ except possibly at $a$.
\end{itemize}
Suppose also that $\varphi : [a,b] \to \mathbb{R}$ is a function satisfying all 
of the following:
\begin{itemize}
\item $\varphi$ is continuous in a neighbourhood of $a$ except possibly at $a$.
\item $\varphi$ admits a power series representation at the left end-point a of 
the form
\begin{align*}
\varphi(x)  \sim \sum_{s=0}^\infty b_s (x - a)^{s + \alpha - 1} 
\end{align*}
for some $\alpha \in \mathbb{C}$ such that $\operatorname{Re} \alpha > 0$, and 
with $b_0 \not = 0$.
\end{itemize}
Furthermore, suppose that there exists $M > 0$ such that for all $\lambda \geq 
M$, the integral $I(\lambda)$ defined by
\begin{align*}
I(\lambda) := \int_{a}^b dx \ \varphi(x) e^{- \lambda h(x)},
\end{align*}
converges absolutely.

Then, as $\lambda\to \infty$,
\begin{align*}
I(\lambda) \sim e^{- \lambda h(a)} \sum_{s=0}^\infty \Gamma \left( \frac{s + 
\alpha}{\mu} \right) \frac{c_s}{\lambda^{\frac{s + \alpha}{\mu}}} \ ,
\end{align*}
where the coefficients $c_s$ are expressible in terms of $a_s$ and $b_s$, and, 
in particular, we have
\begin{align*}
c_0 = \frac{b_0}{\mu a_0^{\frac{\alpha}{\mu}}}.
\end{align*}
\end{theorem}
\begin{proof}
The proof is given \cite{1989}, chapter 2 ``Classical Procedures'', section 1 
``Laplace's method''.
\end{proof}

The previous theorem can be applied to all the Laplace-type integrals that will 
be used in this paper. To be explicit, the most typical usage of this theorem 
will be for the case where $h : [a,c] \to \mathbb{R}$ is a twice continuously 
differentiable strictly convex function, which implies that $h''(x) > 0$ for all 
$x \in [a,c]$. If there exists $b \in (a,c)$ such that $h'(b) = 0$, then this 
point $b$ is the global minimum of $h$ and one can consider the function $h$ on 
the intervals $[a,b]$ and $[b,c]$. Note that the previous theorem holds 
precisely for $h$ on the interval $[b,c]$ since $h$ attains its global minimum 
at the left-end point $b$. For the interval $[a,b]$, one instead considers the 
mapping $\tilde{h}(x) := h(-x)$ defined on the interval $[-b, -a]$. One finds 
that the mapping $\tilde{h}$ obtains a global minimum at $-b$ and again the 
contents of the previous theorem hold.
 
The role of the mapping $\varphi : [a,c] \to \mathbb{R}$ does not change. In 
particular, if $\varphi$ admits a power series representation at any point on 
this interval, then it necessarily also admits power series representations when 
using one sided limits.
 
In the case of a strictly convex $h$, at the global minimum $b$, we have $\mu = 
2$ and $a_0 = \frac{1}{2} h''(b) \not = 0$. The mapping $\varphi$ is of more 
importance. In particular, suppose that $\varphi$ is a smooth function such that 
for some finite $i \in \mathbb{N}$ and for all $k < i$, we have $\varphi^{(k)} 
(b) = 0$ and $\varphi^{(i)} (b) \not = 0$.  In the notation of the previous 
theorem, this would correspond to the situation where $\alpha = i + 1$ and $b_0 
= \frac{1}{i!} \varphi^{(i)} (b)$. 
Applying the previous theorem, we then would have
\begin{align*}
\int_{a}^c dx \ \varphi(x) e^{- \lambda h(x)} \sim e^{- \lambda h(b)} \Gamma 
\left( \frac{i + 1}{2} \right) \frac{1}{i!} \varphi^{(i)} (b) \frac{1}{2 \left( 
\frac{1}{2} h''(b) \right)^{\frac{i+1}{2}}} \frac{1}{\lambda^{\frac{i + 1}{2}}} 
,
\end{align*}
and
\begin{align*}
\frac{\int_{a}^c dx \ \varphi(x) e^{- \lambda h(x)}}{\int_{a}^c dx \ e^{- 
\lambda h(x)}} \sim \frac{\Gamma \left( \frac{i + 1}{2} \right)}{\Gamma \left( 
\frac{1}{2} \right)} \frac{1}{i!} \varphi^{(i)} (b) \frac{1}{ \left( \frac{1}{2} 
h''(b) \right)^{\frac{i}{2}}} \frac{1}{\lambda^{\frac{i}{2}}} .
\end{align*}

The primary message from this is that the order of the first non-zero derivative 
of $\varphi$ determines the rate of vanishing of such Laplace-type integrals, in 
particular, in this case we would have
\begin{align*}
\frac{\int_{a}^c dx \ \varphi(x) e^{- \lambda h(x)}}{\int_{a}^c dx \ e^{- 
\lambda h(x)}} = O ( \lambda^{- \frac{i}{2}}) \,, \quad \lambda\to \infty.
\end{align*}
Furthermore, there will be some cases in which $h'(x) \not = 0$ for any $x \in 
(a,c)$. In such a situation, it will also typically be so that $h'(x) < 0$ for 
all $x \in (a,c)$, this implies that $h$ is minimized at the right end point 
$c$, and, by considering the mapping $\tilde{h}(x) = h(-x)$ again, we see that 
the mapping $\tilde{h}$ is now minimized at its left end point $-c$ and thus the 
theorems hold again. In such a situation, we have $\mu = 1$ instead of $\mu = 
2$.

\end{document}